\documentclass[aps,prx,twocolumn,showpacs,amsmath,notitlepage,amssymb,superscriptaddress]{revtex4-1}

\usepackage{bm}        
\usepackage{amssymb}   
\usepackage{amsmath}
\usepackage{amsthm}
\usepackage{graphicx}  
\usepackage{slashbox}
\usepackage{color}

\definecolor{dartmouthgreen}{rgb}{0.05, 0.5, 0.06}

\newcommand{\bra}[1]{\langle {#1} |}
\newcommand{\ket}[1]{| {#1} \rangle}

\newtheorem{theorem}{Theorem}
\newtheorem{proposition}{Proposition}
\newtheorem{corollary}{Corollary}

\begin{document}

\title{Entanglement-assisted classical communication can simulate classical \\communication without causal order}
\author{Seiseki Akibue}
\email{akibue.seiseki@lab.ntt.co.jp}
\affiliation{Department of Physics, Graduate School of Science, The University of Tokyo 7-3-1 Hongo, Bunkyo-ku, Tokyo 113-0033, JAPAN}
 \affiliation{NTT Communication Science Laboratories, NTT Corporation 3-1 Morinosato Wakamiya, Atsugi-shi, Kanagawa 243-0124, JAPAN}
 
\author{Masaki Owari}
\email{masakiowari@inf.shizuoka.ac.jp}
\affiliation{Department of Computer Science, Faculty of Informatics, Shizuoka University, 3-5-1 Johoku, Naka-ku, Hamamatsu 432-8011, JAPAN}

\author{Go Kato}
\email{kato.go@lab.ntt.co.jp}
 \affiliation{NTT Communication Science Laboratories, NTT Corporation 3-1 Morinosato Wakamiya, Atsugi-shi, Kanagawa 243-0124, JAPAN}

\author{Mio Murao}
\email{murao@phys.s.u-tokyo.ac.jp}
\affiliation{Department of Physics, Graduate School of Science, The University of Tokyo 7-3-1 Hongo, Bunkyo-ku, Tokyo 113-0033, JAPAN}
\date{\today}

\begin{abstract}
Phenomena induced by the existence of entanglement, such as nonlocal correlations, exhibit characteristic properties of quantum mechanics distinguishing from classical theories. When entanglement is accompanied by classical communication, it enhances the power of quantum operations jointly performed by two spatially separated parties.  Such a power has been analyzed by the gap between the performances of joint quantum operations implementable by local operations at each party connected by classical communication with and without the assistance of entanglement.  In this work, we present a new formulation for joint quantum operations connected by classical communication beyond special relativistic causal order but without entanglement and still within quantum mechanics.   Using the formulation, we show that entanglement assisting classical communication necessary for implementing a class of joint quantum operations called \textit{separable operations} can be interpreted to simulate ``classical communication'' not respecting causal order.  Our results reveal a new counter-intuitive aspect of entanglement related to spacetime.
\end{abstract}

\maketitle

\section{Introduction}
Entanglement induces many counter-intuitive phenomena that cannot be described by classical physics, such as the existence of the {\it nonlocal correlations} formulated by Bell and CHSH \cite{Bell}, where entanglement can be interpreted to enhance correlation in space.
However, such correlations produced by entanglement cannot be used for superluminal communication between two parties \cite{PRbox}. For communication, two parties have to be  within a distance across which light can travel, namely, {\it timelike} separated in both quantum and classical cases.

The power of entanglement in relation to time also arises when it is accompanied by classical communication. Quantum teleportation \cite{teleportation},  a protocol for transmitting quantum information from a party (sender) to another timelike separated party (receiver) by using shared entanglement and classical communication from the sender to the receiver, is one example.  Although quantum communication or quantum communication from a mediating third party to two parties is necessary to share a fixed entangled state, this event can be accomplished before the event where the sender \textit{decides} what quantum information to send.  In contrast, the event of the decision should be made before the event of quantum communication in direct quantum communication.  Thus we can slightly ``dodge'' the time-line of the event of quantum communication when entanglement and classical communication are used for transmitting quantum information.

In this paper, we show a new aspect of a power of entanglement accompanied by classical communication. We consider a class of joint quantum operations of two parties not sharing entanglement but communicating by fictatious {\it ``classical communication" without predefined causal order}. We show that the class coincides with a well known subclass of joint quantum operations of two parties sharing entanglement and communicating by (normal) classical communication. This can be interpreted as a new counter-intuitive effect of entanglement shared between the parties assisting classical communication that simulates the ``classical communication" without predefined causal order.

We define ``classical communication" without predefined causal order by extending the formulation of a special class of deterministic joint quantum operations implementable without entanglement between the parties called {\it local operations and classical communication} (LOCC) \cite{VVedral, MBPlenio, Horodecki, CLMOW12}, which is widely used in the field of quantum information for investigating entanglement and nonlocal properties. 
In LOCC, each party performs quantum operations and quantum operations are connected by communication channels. Each quantum operation is well localized in a spacetime coordinate, and thus called a {\it local operation}. We can define the partial order of the spacetime coordinates of the local operations following special relativity, which is referred to as {\it causal order} \cite{OFC, Chiribella1, Chiribella2}. Following the causal order, local operations performed by different parties can be connected by classical communication, whereas local operations within each party can be connected by quantum communication. 

In an attempt to merge quantum mechanics and general relativity, however, the existence of the causal order of local operations in spacetime may not be a fundamental requirement of nature.   It has been recently shown that interesting phenomena arise if the partial order of local operations constituting joint quantum operations is not predefined \cite{OFC, Chiribella1, Chiribella2}.
Alternative representations of communication that are not based on predefined causal order have been proposed \cite{OFC, Hardy}.

We extend LOCC in such a way that we include joint quantum operations between two parties implemented without shared entanglement but with ``classical communication'' without predefined causal order denoted by CC* connecting local operations performed by different parties. We name a new class of deterministic joint quantum operations obtained by this extension but still within quantum mechanics as LOCC*.  
First, we show that LOCC* with CC* respecting causal order is reduced to LOCC.  
Second, we show that LOCC* can be represented by the $\infty$-shaped loop shown in Fig.\ref{fig:loop}.  We also show that it is equivalent to a certain class of deterministic joint quantum operations known as {\it separable operations} \cite{Rains, 9state}, which has been introduced for mathematical simplicity to analyze nonlocal quantum tasks in place of LOCC.  
Third, by considering the correspondence between LOCC* and a probabilistic version of LOCC called stochastic LOCC, we analyze the power of CC* in terms of enhancing the success probability of probabilistic operations in stochastic LOCC. We also investigate the relationship between LOCC* and the process matrix formalism for joint quantum operations without predefined causal order developed in \cite{OFC}. Note that CC* does not contain classical communication to the past that is inconsistent with local operations in LOCC*.

Similar to LOCC, separable operations cannot create entanglement between two parties initially sharing no entanglement. However, some separable operations are {\it not} implementable by LOCC \cite{9state, JNiset, DiVincezo, EChitambar}. Moreover, there exists a separable operation not implementable by LOCC even when infinitely many rounds of classical communications are allowed \cite{CLMOW12}. Since any deterministic joint quantum operations can be implemented by two parties initially sharing entanglement and performing an LOCC protocol, shared entanglement is necessary to implement  non-LOCC separable operations. On the other hand, the ``classical communication" not respecting causal order is necessary to implement the non-LOCC  separable operations without the help of entanglement.  
Combining these two facts, we see that an operational interpretation of separable operations in terms of CC* is given. This interpretation suggests that for implementing non-LOCC separable operations, entanglement assisted classical communication (respecting causal order) simulates ``classical communication'' not respecting causal order.

This paper is organized as follows: In Section II, we first review the formulation of LOCC and then extend it for introducing and defining LOCC* and CC*. We also show that LOCC* can be represented by the $\infty$-shaped loop. In Section III, we construct an example of LOCC* not in LOCC and show that LOCC* is equivalent to the class of separable operations. In Section IV, we present the relationship between LOCC* and stochastic LOCC. In Section V, we investigate the relationship between LOCC* and the process matrix. The last section is devoted to conclusion.

\section{Formulation of LOCC*}

To introduce LOCC*, we first review the formulation of quantum operations and LOCC. In this paper, we consider only finite dimensional quantum systems.   A quantum operation is the most general map transforming an input to an output.    We can consider both quantum states and classical bits as either the input or output of a quantum operation.  The most general quantum operation describes a probabilistic situation where a transformation depending on a classical input is applied to an arbitrary quantum state and then a classical output and a quantum output are probabilistically obtained.   The probabilistic transformation can be interpreted as a quantum measurement, and the classical output corresponds to the outcome of the measurement. A deterministic quantum operation can be understood as a special case of a general quantum operation, which only takes the quantum output. We can also obtain a deterministic quantum operation from a probabilistic quantum operation by averaging over all possible probabilistic quantum transformations with a weight determined by the probability of the classical output.  This averaging process can be  considered as discarding information regarding the classical output.

Mathematically, a general quantum operation can be represented by a {\it quantum instrument} conditioned by a classical input $i$, which is described by a family of linear maps $ \{ \mathcal{M}_{o|i} \}_o$ that transforms a quantum input state $ \rho$ to a quantum output state given by $\mathcal{M}_{o|i} ( \rho) / p(o|i) $ associated with a classical output $o$ ($o=1,2,\cdots, n$) with a conditional probability  distribution $p(o|i)={\rm tr}[\mathcal{M}_{o|i} ( \rho) ]$.  Each element of the instrument $\mathcal{M}_{o|i}$ has to be a {\it completely positive} (CP) map to describe quantum operations allowed in quantum mechanics.  The CP property of the map is required to enable a quantum operation to be performed only on a subsystem of a composite system while keeping the total state of the composite system a valid quantum state. 
Since the average of all transformations $\sum_o p(o|i)\cdot \mathcal{M}_{o|i}/p(o|i)=\sum_o\mathcal{M}_{o|i}$ describes a deterministic quantum operation, it has to be a {\it trace preserving} (TP) map, that is  the probability of obtaining the quantum output of the averaged map is given by ${\rm tr} \left[\sum_o \mathcal{M}_{o|i} ( \rho)\right] ={\rm tr} \left[ \rho\right]= 1$.
A quantum operation without classical input is a quantum instrument denoted by $\{ \mathcal{M}_o \}_o$. Note that the subscripts of a family represent classical outputs, over which sum of elements in the family is a CPTP map. A quantum operation without classical output is  a quantum instrument whose element is a CPTP map, and it is denoted simply by $\mathcal{M}_{|i}$ or $\mathcal{M}$ (instead of $\{\mathcal{M}_{|i}\}$ or $\{\mathcal{M}\}$).
Graphical representations of these four types of quantum operations ($\{\mathcal{M}_{o|i}\}_o$, $\{\mathcal{M}_{o}\}_o$, $\mathcal{M}_{|i}$ and $\mathcal{M}$) are  given in Fig.\ref{fig:instruments}.  

\begin{figure}
 \centering
  \includegraphics[height=.12\textheight]{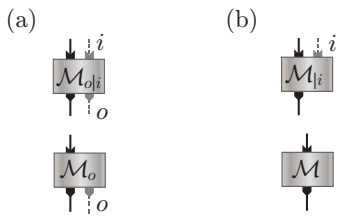}
  \caption{ Graphical representations of  four types of quantum operations, $\{\mathcal{M}_{o|i}\}_o$, $\{\mathcal{M}_{o}\}_o$, $\mathcal{M}_{|i}$ and $\mathcal{M}$. A box represents a quantum operation. The upper terminals of a box represent inputs and the lower ones represent outputs. Dotted lines represent classical input or classical output and solid lines represent quantum one. (a) Probabilistic operations have a classical output corresponding to the outcome of the measurement.  (b) Deterministic operations do not provide a classical output.}
\label{fig:instruments}
\end{figure} 

We consider a deterministic joint quantum operation (a CPTP map) taking only the quantum input and output denoted by $\mathcal{M}$ implemented by two parties, Alice and Bob,  who are connected only by classical communication.  We set Alice to perform the first operation.   As the simplest case, we consider linking a classical output  $o$ of Alice's local operation represented by $\{ \mathcal{A}_o \}_o$ and a classical input $i$ of Bob's local operation represented by $\mathcal{B}_{|i}$.   Since Alice and Bob are acting on different quantum systems at different spacetime coordinates,  the joint quantum operation is described by a tensor product of two local operations.  The classical communication ability indicates that the spacetime coordinate of Alice's local operation and that of Bob's local operation are timelike separated, and Bob's coordinate is in the future cone of Alice's coordinate.   Linking the classical output of Alice $o$ and the classical input of Bob $i$ means that they are perfectly correlated, namely setting $i=o$ for all $o$.   By taking averages over $o$ and $i$,  we obtain a deterministic joint quantum operation given by 
\begin{equation}
\mathcal{M}=\sum_{o,i}\delta_{i,o} \mathcal{A}_o\otimes\mathcal{B}_{|i}=\sum_{o} \mathcal{A}_o\otimes\mathcal{B}_{|o},
\label{eq:LOCC1}
\end{equation}
where  $\delta_{o,i}$ denotes the Kronecker delta.  This is the simplest case of LOCC, which is called one-way LOCC. 

A deterministic joint operation represented by more general  finite-round LOCC between two parties is defined by  connecting a sequence of Alice's local operations given by $\{ \mathcal{A}^{(N)}_{o_N|i_N}\circ\cdots\circ \mathcal{A}^{(1)}_{o_1} \}_{o_N,\cdots,o_1}$ and a sequence of Bob's local operations  given by $\{\mathcal{B}^{(N)}_{o'_N|i'_N}\circ\cdots\circ \mathcal{B}^{(1)}_{o'_1|i'_1} \}_{o'_N,\cdots,o'_1}$.   Here $\circ$ denotes a connection between two local operations at timelike separated coordinates linking the quantum output of a local operation and the quantum input of the next local operation of the {\it same} party and introduces a total order for each party's local operations.  The indices $i_k$ and $i'_k$ are classical inputs of the $k$-th operations and $o_k$ and $o'_k$ are classical outputs of the $k$-th operations of Alice and Bob, respectively.   

In LOCC, the local operations within the parties are totally ordered, and the spacetime coordinates of all the local operations of Alice and Bob are totally ordered alternately.   Then $o_1$ of Alice's local operation and $i'_1$ of Bob's local operation are linked by classical communication and setting $i'_1=o_1$, and similarly, $o'_1$ of Bob's local operation and $i_{2}$ of Alice's local operation are linked by setting $i_{2}= o'_1$ and so on.  Thus,  finite-round LOCC between the two parties is defined by a set of deterministic joint quantum operations represented by
\begin{eqnarray}
\mathcal{M}&=&\sum_{i_1,\cdots,i'_N,o_1,\cdots, o'_N} 
\delta_{o_N,i'_N} \cdots \delta_{o_1,i'_1}\delta_{i_1,1}
\nonumber\\&& \mathcal{A}^{(N)}_{o_N|i_N}\circ\cdots\circ \mathcal{A}^{(1)}_{o_1|i_1}\otimes \mathcal{B}^{(N)}_{o'_N|i'_N}\circ\cdots\circ \mathcal{B}^{(1)}_{o'_1|i'_1} 
\label{eq:LOCC2} \\
&=& \sum_{} \mathcal{A}^{(N)}_{o_N|i_N}\circ\cdots\circ \mathcal{A}^{(1)}_{o_1|1}\otimes\mathcal{B}^{(N)}_{o_N'|o_N}\circ\cdots\circ \mathcal{B}^{(1)}_{i_2|o_1}
\end{eqnarray}
 where we define $\mathcal{A}^{(1)}_{o_1|1} := \mathcal{A}^{(1)}_{o_1}$. An example for $N=3$ is shown in Fig.\ref{fig:twoCC}.  

\begin{figure}
 \centering
  \includegraphics[height=.30\textheight]{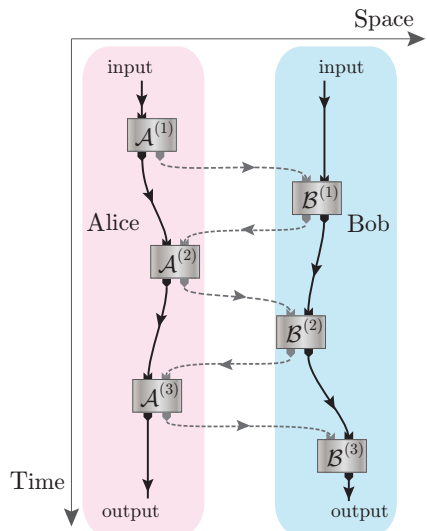}
  \caption{Joint quantum operations of two parties, Alice and Bob, consisting of local operations. Dotted arrows represent classical communication between local operations and solid arrows represent quantum communications between local operations. In an LOCC protocol, local operations are  totally ordered.  Local operations performed by different parties are connected only by classical communication, whereas local operations within each party are connected by quantum communication. }
\label{fig:twoCC}
\end{figure} 

Now we develop a framework to investigate local operations and classical communication between the parties without the assumption of the existence of the predefined ordering of the spacetime coordinates.   That is, we keep the local spacetime coordinates  and the totally ordered structure of the local operations within each party but we do not assign the {\it global} spacetime coordinate across the parties. Moreover, we allow the connection of any classical inputs and outputs of local operations between different parties as long as the resulting deterministic joint quantum operation falls within the scope of quantum mechanics. This relaxation allows the generalization of classical communication to a conditional probability distribution $p(i_1,\cdots,i'_N|o_1,\cdots,o'_N)$ linking the classical outputs and classical inputs of local  operations.
The  corresponding joint quantum operation is represented by
\begin{eqnarray}
\mathcal{M}&=&\sum_{i_1,\cdots,i'_N,o_1,\cdots,o'_N} p(i_1,\cdots,i'_N|o_1,\cdots,o'_N)\nonumber\\
&&\mathcal{A}^{(N)}_{o_N|i_N}\circ\cdots\circ \mathcal{A}^{(1)}_{o_1|i_1}\otimes \mathcal{B}^{(N)}_{o'_N|i'_N}\circ\cdots\circ \mathcal{B}^{(1)}_{o'_1|i'_1} .
\label{eq:LOCC*0}
\end{eqnarray}
This generalization does not guarantee that the joint quantum operation represented by Eq.~\eqref{eq:LOCC*0} is TP  whereas its CP property is preserved.   Since we are investigating deterministic joint quantum operations, we require $p(i_1,\cdots,i'_N|o_1,\cdots,o'_N)$ to keep the form of Eq.~\eqref{eq:LOCC*0} to represent a CPTP map.   We refer to a set of CPTP maps in the form of  Eq.~\eqref{eq:LOCC*0}  with $p(i_1,\cdots,i'_N|o_1,\cdots,o'_N)$  as LOCC*.  

For one-way LOCC,  such a generalization corresponds to replacing the delta function $\delta_{o,i}$ in Eq.~\eqref{eq:LOCC2} with a conditional probability distribution $ p(i|o) $. This is equivalent to replacing a perfect classical channel with a general noisy classical  channel. However, this is not the case for multiple-round LOCC.  As the first result,  we show that a joint quantum operation is in LOCC if and only if it can be decomposed in the form of Eq.~(\ref{eq:LOCC*0}) with $p(i_1,\cdots,i'_N|o_1,\cdots,o'_N)$ respecting the {\it causal order} of the classical inputs and outputs of local operations imposed by special relativity, namely, the causal order determined by the partial order structure of the spacetime where local operations are performed. Rigorous definitions and the proof are given in Appendix A.

It is easy to check that without loss of generality, any operations in LOCC* can be implemented by just one quantum operation performed by each party connected by a conditional probability distribution, since by letting $o_A:=(o_1,\cdots,o_N)$, $i_A:=(i_1,\cdots,i_N)$, $o_B:=(o_1',\cdots,o_N')$ and $i_B:=(i_1',\cdots,i_N')$, we can regard the sequence of Alice's local operations $\{ \mathcal{A}^{(N)}_{o_N|i_N}\circ\cdots\circ \mathcal{A}^{(1)}_{o_1|i_1} \}_{o_N,\cdots,o_1}$ as one quantum instrument $\{\mathcal{A}_{o_A|i_A}\}_{o_A}$ conditioned by the classical input $i_A$, as with Bob's local operations and $p(i_A,i_B|o_A,o_B)$ is still a conditional probability distribution.   

Thus, LOCC* is simply defined by a set of CPTP maps $\mathcal{M}$ given in the form of
\begin{eqnarray}
\mathcal{M}=\sum_{i_A,i_B,o_A,o_B}p(i_A,i_B|o_A,o_B)\mathcal{A}_{o_A|i_A}\otimes\mathcal{B}_{o_B|i_B},
\label{eq:LOCC*}
\end{eqnarray}
where $p(i_A,i_B|o_A,o_B)$ is a  conditional probability distribution, $\{ \mathcal{A}_{o_A|i_A} \}_{o_A}$ is Alice's quantum operation with a classical input $i_A$ and a classical output $o_A$, and $\{ \mathcal{B}_{o_B|i_B} \}_{o_B}$ is Bob's quantum operation with a classical input $i_B$ and a classical output $o_B$.
Note that the quantum operation performed by each party is not necessary to be localized in a spacetime coordinate, however, we refer to them as {\it local} operations to distinguish them from a {\it joint} operation.
We refer to $p(i_A,i_B|o_A,o_B)$ keeping the form of Eq.~\eqref{eq:LOCC*} to be a CPTP map as CC*, namely {\it ``classical communication'' without predefined causal order}, with respect to local operations $\{\mathcal{A}_{o_A|i_A}\}_{o_A}$ and $\{\mathcal{B}_{o_B|i_B}\}_{o_B}$.   Note that we can also ``collapse'' the sequential local operations in  multi-round LOCC to represent a joint quantum operation in the form of Eq.~\eqref{eq:LOCC*}.  However, in this case if we regard the combined  operations $\{\mathcal{A}_{o_A|i_A}\}_{o_A}$ and $\{\mathcal{B}_{o_B|i_B}\}_{o_B}$ as operations localized in the spacetime, the  corresponding $p(i_A,i_B|o_A,o_B)$ cannot be interpreted as classical communication respecting causal order.

\begin{figure}
 \centering
  \includegraphics[height=.20\textheight]{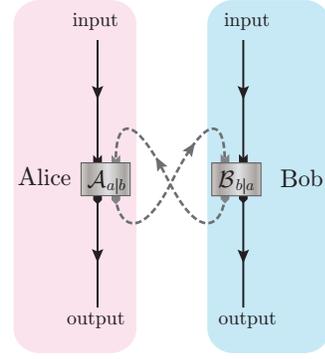}
  \caption{A deterministic joint quantum operation consisting of local operations connected by the $\infty$-shaped loop. Any element of LOCC* can be represented by this joint quantum operation and vice versa. }
\label{fig:loop}
\end{figure} 

As the second result, LOCC* can always be represented by a $\infty$-shaped  loop as shown in Fig.\ref{fig:loop}. Namely, we can show the following theorem.
\begin{theorem}
 $\mathcal{M}$ is LOCC* if and only if $\mathcal{M}$ is a CPTP map and can be decomposed into the form of
 \begin{eqnarray}
\mathcal{M}=\sum_{a,b}\mathcal{A}_{a|b}\otimes\mathcal{B}_{b|a},
\label{eq:LOSC}
\end{eqnarray}
where $\{\mathcal{A}_{a|b}\}_a$ and $\{\mathcal{B}_{b|a}\}_b$ are local operations performed by Alice and Bob, respectively.
\end{theorem}

\begin{proof}
It is easy to see that if $p(i_A,i_B|o_A,o_B)=\delta_{i_A,o_B}\delta_{i_B,o_A}$, LOCC* defined by Eq.~(\ref{eq:LOCC*})  reduces to Eq.~(\ref{eq:LOSC}).
To show the converse, let the local operations of Alice and Bob be
\begin{eqnarray}
\tilde{\mathcal{A}}_{i_B,x|o_A,o_B}&:=&\sum_{i_A} p(i_A,i_B|o_A,o_B)\mathcal{A}_{x|i_A}\\
\tilde{\mathcal{B}}_{o_A,o_B|i_B,x}&:=&\delta_{x,o_A}\mathcal{B}_{o_B|i_B},
\end{eqnarray}
by introducing a new  index $x$. Since a convex combination of TP maps is also a TP map,  $\{\tilde{\mathcal{A}}_{i_B,x|o_A,o_B}\}_{i_B,x}$ and $\{\tilde{\mathcal{B}}_{o_A,o_B|i_B,x}\}_{o_A,o_B}$ are quantum instruments. Introducing new classical indices  $a:=(i_B, x)$ and $b:=(o_A, o_B)$,  any $\mathcal{M}$ in the form of Eq.~(\ref{eq:LOCC*}) has a decomposition in the form of
\begin{equation}
 \mathcal{M}=\sum_{a,b}\tilde{\mathcal{A}}_{a|b}\otimes\tilde{\mathcal{B}}_{b|a},
\end{equation}
which is the form given by Eq.~(\ref{eq:LOSC}).
\end{proof}

From this form of LOCC*,  it is possible to interpret that CC* in LOCC* can be {\it looped}, namely, Alice's classical input is Bob's classical output and Bob's classical input is Alice's classical output. 

\section{LOCC* and SEP}

We show that LOCC* provides an operational interpretation of a set of {\it separable operations} (denoted by SEP) by proving that LOCC* is equivalent to SEP.  We start by investigating inclusion relations between LOCC and LOCC*.   It is easy to verify that LOCC is a subset of LOCC*  from the definition of LOCC*.   We show that LOCC* is strictly larger than LOCC by constructing the following  example based on nine-state discrimination \cite{9state}.

{\bf Example:} Reference \cite{9state} considers the task of distinguishing a set of nine mutually orthogonal product states of two three-level systems shared between Alice and Bob given by
\begin{align}
&\bigr|\psi_{1\left(2\right)}\bigr\rangle_{AB} =\bigr|0\bigr\rangle_{A}\bigr|0\pm1\bigr\rangle_{B},\,\,\,\,\,\,\,\,\,
\bigr|\psi_{3\left(4\right)}\bigr\rangle_{AB} =\bigr|0\pm1\bigr\rangle_{A}\bigr|2\bigr\rangle_{B},\nonumber \\
&\bigr|\psi_{5\left(6\right)}\bigr\rangle_{AB} =\bigr|2\bigr\rangle_{A}\bigr|1\pm2\bigr\rangle_{B},\,\,\,\,\,\,\,\,\,
\bigr|\psi_{7\left(8\right)}\bigr\rangle_{AB} =\bigr|1\pm2\bigr\rangle_{A}\bigr|0\bigr\rangle_{B},\nonumber\\
&\bigr|\psi_{9}\bigr\rangle_{AB} =\bigr|1\bigr\rangle_{A}\bigr|1\bigr\rangle_{B}\label{eq:nine_states},
\end{align}
where $\{ \ket{j} \}_{j=0}^2 $ is the computational basis of a three-level system,  $\ket{a\pm b} :=(\ket{a}\pm\ket{b})/\sqrt{2}$, indices $A$ and $B$ represent Alice's share and Bob's share of the states, respectively.  It has been shown that a state guaranteed to be one of the set of nine states is {\it not} deterministically and perfectly distinguishable from other states by any protocol described by a map in LOCC followed by local measurements \cite{9state}. 
  We show that the the nine-state  {\it can} be deterministically  and perfectly distinguishable by using a map in LOCC* followed by local measurements by presenting constructions of Alice's local operation $\{\mathcal{A}_{a|b} \}_a$  and Bob's local operation $\{ \mathcal{B}_{b|a} \}_b$  in the form of Eq.~(\ref{eq:LOSC}).   The constructions of $\{\mathcal{A}_{a|b} \}_a$ and  $\{ \mathcal{B}_{b|a} \}_b$ in Kraus operator representations are given in Table~\ref{table:9state}.   It is easy to check that for any $\ket{\psi_k }\in \{  \ket{\psi_i} \}_{i=1}^9$, the corresponding $\mathcal{M}$  in LOCC* with the constructions of the local operations transforms $\ket{\psi_k}\bra{\psi_k}$ into
\begin{eqnarray}
\sum_{a,b}\mathcal{A}_{a|b}\otimes\mathcal{B}_{b|a} (\ket{\psi_k}\bra{\psi_k})
= \ket{k}_{A'}\bra{k} \otimes \ket{k}_{B'} \bra{k}
\end{eqnarray}
where indices $A'$ and $B'$ denote nine-dimensional output systems for Alice and Bob.    Once Alice and Bob obtain the output state $\ket{k}_{A'}\bra{k} \otimes \ket{k}_{B'} \bra{k}$, they can determine the classical output $k$ by individually performing projective measurements in the basis given by $\{ \ket{j} \}_{j=1}^{9}$. Note that the nine states can be probabilistically distinguished without error by a stochastic LOCC (SLOCC) protocol, which indicates that LOCC* is also closely related to SLOCC as shown later.

\begin{table}
\begin{center}
\renewcommand
\arraystretch{1.1}
\begin{tabular}{cp{20mm}p{20mm}p{20mm}}\hline\hline
\backslashbox{a}{b} & 1 & 2 & 3 \\ \hline
1 & $\ket{1}_{A'}\bra{0}_A$ & $\ket{2}_{A'}\bra{0}_A$ & $\ket{3}_{A'}\bra{0+1}_A$ \\ 
2 & $\ket{8}_{A'}\bra{1-2}_A$ & $\ket{9}_{A'}\bra{1}_A$ & $\ket{4}_{A'}\bra{0-1}_A$ \\ 
3 & $\ket{7}_{A'}\bra{1+2}_A$ & $\ket{6}_{A'}\bra{2}_A$ & $\ket{5}_{A'}\bra{2}_A$ \\ \hline\hline
\end{tabular}
\vspace{0.5cm}
\\
\begin{tabular}{cp{20mm}p{20mm}p{20mm}}\hline\hline
\backslashbox{a}{b} & 1 & 2 & 3 \\ \hline
1 & $\ket{1}_{B'}\bra{0+1}_B$ & $\ket{2}_{B'}\bra{0-1}_B$ & $\ket{3}_{B'}\bra{2}_B$ \\ 
2 & $\ket{8}_{B'}\bra{0}_B$ & $\ket{9}_{B'}\bra{1}_B$ & $\ket{4}_{B'}\bra{2}_B$ \\
3 & $\ket{7}_{B'}\bra{0}_B$ & $\ket{6}_{B'}\bra{1-2}_B$ & $\ket{5}_{B'}\bra{1+2}_B$ \\ \hline\hline
\end{tabular}
\caption{ Tables of the Kraus operators $ K_{a|b}^{(A)}$ of Alice's local operations $\{ \mathcal{A}_{a|b} \}_a$ in  the upper table and $K_{b|a}^{(B)}$ of Bob's local operation $\{\mathcal{B}_{b|a} \}_b$ in  the lower table.   In the Kraus operator representation, the deterministic joint quantum operation $\mathcal{M}=\sum_{a,b}\mathcal{A}_{a|b}\otimes\mathcal{B}_{b|a}$ transforms any quantum input $\rho_{AB}$ on systems $A$ and $B$ into a quantum output on systems $A'$ and $B'$ as ${\rho'} _{A' B'}=\mathcal{M}(\rho_{AB}) =   \sum_{a,b} (K_{a|b}^{(A)} \otimes K_{b|a}^{(B)}) \rho_{AB} (K_{a|b}^{(A)}  \otimes K_{b|a}^{(B)})^\dagger$ where $\sum_a (K_{a|b}^{(A)})^\dagger K_{a|b}^{(A)} = \mathbb{I}_A$ for any $b$ and $\sum_b (K_{b|a}^{(B)})^\dagger K_{b|a}^{(B)} = \mathbb{I}_B$ for any $a$ with identity operators $\mathbb{I}_A$ and $\mathbb{I}_B$ on system $A$ and $B$, respectively.}
\label{table:9state} 
\end{center}
\end{table}

\begin{figure}
 \centering
  \includegraphics[height=.20\textheight]{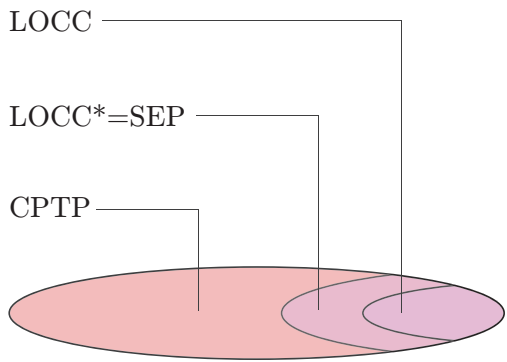}
  \caption{ The inclusion relation between the classes of  deterministic joint quantum operations CPTP, SEP, LOCC*, and LOCC.   LOCC* is equivalent to the set of separable operations (SEP). LOCC is strictly smaller than SEP. LOCC* is strictly smaller than the set of CPTP.}
\label{fig:classes}
\end{figure} 

We can further show that LOCC* is equivalent to SEP as summarized in Fig.~\ref{fig:classes}.  SEP is a set of maps representing deterministic joint quantum operations $\mathcal{M}$ that can be written as
\begin{eqnarray}
 \mathcal{M}=\sum_{k=1}^K\mathcal{E}^A_{k}\otimes\mathcal{E}^B_{k},
\label{eq:SEP}
\end{eqnarray}
where each elements of a quantum operation $\{\mathcal{E}^A_k\otimes\mathcal{E}^B_k\}_k$ is the tensor product of two completely positive maps $\mathcal{E}^A_k$ and $\mathcal{E}^B_k$ performed by Alice and Bob, respectively.
 The set of nine states can be deterministically and perfectly distinguished by using a map in SEP whose Kraus operator representation is given by $\{ \ket{k}_{A'}\otimes \ket{k}_{B'}  \bra{\psi_k}_{AB} \}$ followed by local projective measurements in the basis given by $\{ \ket{k} \}_{k=1}^{9}$.   By definition, a map in SEP cannot transform any separable states into entangled states \cite{VVedral, MBPlenio}.  Therefore, SEP does not have the power to create entanglement between two parties if the quantum input is not entangled.   The class SEP includes the class LOCC \cite{CLMOW12}.   Due to the mathematical simplicity of its structure, the class SEP is often used for proving that a quantum task is not implementable by LOCC protocols by showing that the task is not implementable even by using a stronger class of operations, SEP.  However, its operational meaning has not been clear. 

\begin{theorem}
 LOCC*=SEP
\end{theorem}
\begin{proof}
 It is easy to see that LOCC* is a subset of SEP from the form of Eq.~\eqref{eq:LOSC},  since both $\mathcal{A}_{a|b}$ and $\mathcal{B}_{b|a}$ are completely positive, so we can always transform a map in LOCC* into the form of \eqref{eq:SEP}. 
 To prove that SEP is a subset of LOCC*, we first define a completely positive and trace decreasing (CPTD) map. A linear map $\mathcal{M}$ is CPTD if and only if $\mathcal{M}$ is CP and
 \begin{equation}
 {\rm tr}\left[\mathcal{M}(\rho)\right]\leq {\rm tr}\left[\rho\right]=1
\end{equation}
for all quantum state $\rho$. In Eq.~\eqref{eq:SEP}, we can assume that each $\mathcal{E}^A_k$ and $\mathcal{E}^B_k$ is a CPTD map in general. For any CPTD map $\mathcal{M}$, we can find a complementary CPTD map $\overline{\mathcal{M}}$ such that $\mathcal{M}+\overline{\mathcal{M}}$ is a CPTP map. By using these facts, we present the construction of local operators for LOCC*.
\begin{eqnarray}
\mathcal{A}_{k|k}&=&\mathcal{E}_k^A\,\,\,\,\,\,\mathrm{for}\,\,1\leq k\leq K\\
\mathcal{B}_{k|k}&=&\mathcal{E}_k^B\,\,\,\,\,\,\mathrm{for}\,\,1\leq k\leq K\\
\mathcal{A}_{K+1|k}&=&\overline{\mathcal{E}_k^A}\,\,\,\,\,\,\mathrm{for}\,\,1\leq k\leq K\\
\mathcal{B}_{K+1|k}&=&\overline{\mathcal{E}_k^B}\,\,\,\,\,\,\mathrm{for}\,\,1\leq k\leq K\\
\mathcal{A}_{K+1|K+1}=\mathcal{A}_{K+2|K+2}&=&\mathcal{M}^A\\
\mathcal{B}_{K+1|K+2}=\mathcal{B}_{K+2|K+1}&=&\mathcal{M}^B\\
\mathcal{A}_{a|b}&=&0\,\,\,\,\,\,\,\mathrm{else},\\
\mathcal{B}_{b|a}&=&0\,\,\,\,\,\,\, \mathrm{else},
\end{eqnarray}
where $\mathcal{M}^A$ and $\mathcal{M}^B$ are arbitrary CPTP maps performed by Alice and Bob, respectively.
When the summations of the local operators are taken to form an element of LOCC* as in Eq.~\eqref{eq:LOSC}, the indices $a$ and $b$ run from $1$ to $K+2$.

\end{proof}
Note that the LOCC* map obtained in the proof is different from the simpler LOCC* map given in Table~\ref{table:9state} for nine-state discrimination. 

Since any deterministic joint quantum operations can be implemented by entanglement assisted LOCC, LOCC* can be implemented by entanglement assisted (normal) classical communication respecting causal order.  Combining the results shown in this part, we see that an operational interpretation of SEP in terms of CC* is given.  This interpretation suggests that for implementing maps in SEP but not in LOCC, entanglement assisted classical communication simulates a special class of CC* not respecting causal order, namely, {\it ``classical communication'' without causal order}.

\section{LOCC* and SLOCC}
In this part, we  present the relationship between LOCC* and SLOCC as the third result.  SLOCC is a set of (linear) CP maps consisting of local operations and normal classical communication.   In contrast to LOCC, SLOCC contains CP maps representing cases where particular measurement outcomes are post-selected.   Since we use a linear map to represent SLOCC, a SLOCC element is not always TP but can be trace decreasing (TD) in general.  We define a class of linear CP maps called SLOCC* that can be simulated by SLOCC.
That is, SLOCC* is a set of linear CP maps satisfying $\mathcal{M}=c\Gamma$, where $\Gamma$ is an element of SLOCC and a non-negative constant $c \geq 0$.
Note that SLOCC* contains not only TD maps but also trace increasing maps.  It is easy to see that LOCC* (or SEP) is a subset of SLOCC* since any element of LOCC* (or SEP) is implementable by SLOCC with a constant  success probability independent of inputs.  

 In the following, we show that a superset of LOCC* where the TP condition is removed from LOCC* is equivalent to SLOCC*.  The formal statement is represented as the following theorem.  
 
 \begin{theorem}
 SLOCC* is equivalent to the set of linear CP maps that can be decomposed into the form of Eq.~\eqref{eq:LOCC*}.
\end{theorem}
\begin{proof}
 By definition, SLOCC* is equivalent to the set of linear CP maps that can be decomposed into the form of Eq.~\eqref{eq:SEP}. Without loss of generality, we can restrict $\mathcal{E}_k^A$ and $\mathcal{E}_k^B$ in Eq.~\eqref{eq:SEP} to CPTD maps since for all linear CP maps $\mathcal{E}^A_k$, there is a natural number $M$ such that  $\frac{1}{M}\mathcal{E}^A_k$ is a CPTD map and 
$\mathcal{M}$ given by Eq.~\eqref{eq:SEP} can be represented by $\sum_{i=1}^M\sum_{k}\tilde{\mathcal{E}^A_{i,k}}\otimes\tilde{\mathcal{E}^B_{i,k}}$, where $\tilde{\mathcal{E}^A_{i,k}}=\frac{1}{M}\mathcal{E}^A_k$ and $\tilde{\mathcal{E}^B_{i,k}}=\mathcal{E}^B_{k}$. By employing the same technique used in the proof of Theorem 2, any element of SLOCC* can be decomposed into  the form of Eq.~\eqref{eq:LOSC}. Thus, any SLOCC* element can be decomposed into Eq.~\eqref{eq:LOCC*}.
\end{proof}
We summarize the inclusion relation of sets of linear CP maps as shown in Fig.~\ref{fig:classes2}.

\begin{figure}
 \centering
  \includegraphics[height=.16\textheight]{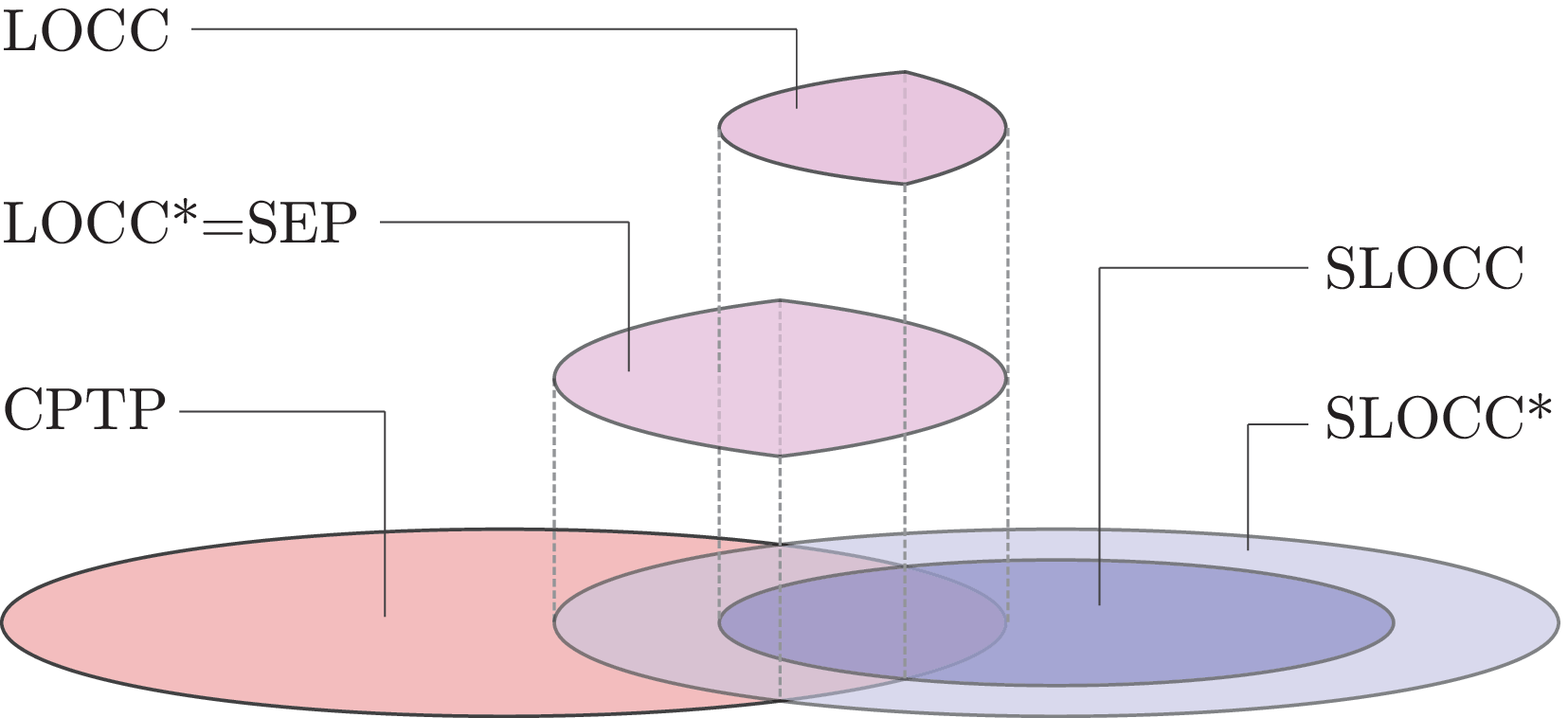}
  \caption{The inclusion relation between the classes of linear CP maps. LOCC* is equivalent to the set of separable maps (SEP). LOCC is strictly smaller than SEP. SEP is strictly smaller than the set of CPTP. The intersection of SLOCC and the set of CPTP maps (TP SLOCC elements) is LOCC since any TP SLOCC element can be implemented by a LOCC protocol. The intersection of SLOCC* and the set of CPTP maps (TP SLOCC* elements) is LOCC*. Note that SEP is not included by SLOCC in this inclusion relation while every map in SEP can be simulated by SLOCC. A SLOCC map whose success probability is strictly less than 1 is defined by a linear CPTD map and thus not included in the set of CPTP map.  The success probability of simulating maps in SEP but not in LOCC by SLOCC is strictly less than 1.}
\label{fig:classes2}
\end{figure} 

Any linear CP maps that can be decomposed into the form of Eq.~\eqref{eq:LOCC*} are simulatable by LOCC with post-selection (SLOCC), and vice versa. A connection between post-selection and the causal structure of the spacetime has been discussed in the context of the probabilistic closed time-like curve (p-CTC) in  \cite{pCTC1, pCTC2, pCTC3}.   Our result indicates the connection between post-selection and CC* without causal order.    In standard spacetime,  a map in SEP but not in LOCC can be implemented by LOCC assisted by entanglement.  Thus we can recognize that entanglement accompanied with classical communication is used only for enhancing the success probability of a task achievable without entanglement $p < 1$ to $p=1$ to implement a map in SEP (=LOCC*,  but not in LOCC).  In this sense,  CC* without causal order can be understood as an alternative characterization of the power of entanglement to enhance success probability by effectively changing the partial ordering properties of the spacetime.

\section{LOCC* and classical process matrix}
We discuss the correspondences between our formalism of LOCC* and other formalisms of quantum operations without the assumptions of  the partial order of local operations, namely {\it higher order formalisms} developed by \cite{OFC, Chiribella2}.   In  the higher order formalisms, the effects of quantum communication linking the local operations of two parties are described as a map that  transforms local operations into a deterministic joint quantum operation.   The map, represented by  a {\it process matrix},  is a higher order map (supermap) transforming a quantum operation to another quantum operation, whereas a (normal) map  transforms a quantum state to another quantum state.   Reference \cite{OFC} derives the requirements for a process matrix to be consistent with quantum mechanics but without a predefined causal order of local operations.  They have shown that there are process matrices that are not implementable by {\it quantum} communication linking partially ordered local operations.  These requirements for a process matrix to be consistent with quantum mechanics can be interpreted as a new kind of causality, which is different from the special relativistic causality but based only on quantum mechanics.  The process matrix shows a new possibility for increasing the speed of quantum computers \cite{Qswitchspeed}, and implementations of the process matrix are discussed for several settings \cite{Qswitch, Nakago}.

A crucial difference between the higher order formalisms and our formalism of LOCC* is that the local operations are linked by quantum communication in the higher order formalisms,  but we only allow classical communication between the parties.   To compare the two formalisms,  we consider a special type of process matrix, called {\it classical} process matrix, where quantum communication between the parties is restricted to transmitting a probabilistic mixture of ``classical'' states, namely, a set of fixed mutually orthogonal states.  In \cite{OFC},  it is shown that a deterministic joint operation described by local operations linked by a classical process matrix does not exhibit the new causality exhibited by a (fully quantum) process matrix. In Appendix B, we show that the deterministic joint operations in this case reduce to a probabilistic mixture of two types of operations in one-way LOCC from Alice to Bob and from Bob to Alice.   We denote a set of such deterministic joint quantum operations as LOCC** for comparison with LOCC*.

Since LOCC** is a set of probabilistic mixtures of one-way LOCC, LOCC* is a larger set than LOCC**.  Therefore, CC* used in implementing non-LOCC quantum operations cannot be represented by a classical process matrix linking two local operations.   The gap between LOCC* and LOCC** originates from the conditions imposed on local operations to restrict CC* and a classical process matrix.  CC* in LOCC* is only required to guarantee the joint quantum operation to be deterministic for {\it some} choices of local operations, whereas a classical process matrix in LOCC** is required to guarantee that the joint quantum operation is deterministic for {\it arbitrary} choices of  local operations.  Hence CC* is less restricted than a classical process matrix.   Considering that CC* is simulated by entanglement assisted classical communication, we can conclude that entanglement provides the power to waive restrictions on classical communication linking local operations originating from both the causality in special relativity and the restriction for classical process matrices when it is accompanied by classical communication.

 \section{Conclusion}
We have developed a framework to describe deterministic joint quantum operations of two parties where the parties do not share entanglement and perform local operations described by normal quantum mechanics and linked by a conditional probability distribution which satisfies the TP property of the joint quantum operation, called “classical communication” without predefined causal order (CC*).
By using the framework, we have first shown that LOCC* with CC* respecting causal order is equivalent to LOCC.
Second, we have shown that LOCC* is equivalent to SEP.
Third, we have shown that non-deterministic LOCC* is equivalent to a set of joint quantum operations that can be simulated by SLOCC (SLOCC*).
We have also shown that LOCC* contains a set of deterministic joint quantum operations where local operations are linked by a classical process matrix (LOCC**).
We remark that the same relation holds for a multipartite setting.
Our framework gives a new interpretation of SEP and would be a new toolbox for analyzing the gap between SEP and LOCC.

\vspace{1cm}

\begin{acknowledgments}
This work is supported by the Project for Developing Innovation Systems of MEXT, Japan and JSPS by KAKENHI (Grant No.~23540463, 2633006,15H01677, 16H01050, 17H01694).  We also gratefully acknowledge the ELC project (Grant-in-Aid for Scientific Research on Innovative Areas MEXT KAKENHI (Grant No.~24106009)) for encouraging this research.
\end{acknowledgments}

\appendix

\renewcommand{\theequation}{A\arabic{equation}}
 \section{LOCC and CC* respecting causal order}

In Appendix A, we analyze a set of joint quantum operations represented by Eq.\eqref{eq:LOCC*0} with CC* respecting the causal order. We assume that the spacetime coordinates of the local operations are {\it partially} ordered and ones within each parties are {\it totally} ordered. We consider $p(i_1,\cdots,i'_N|o_1,\cdots,o'_N)$ consistent with some partial order of local operations, which is called CC* respecting causal order. CC* respecting causal order contains noisy classical channels (in the presence of memory) connecting inputs and outputs of several local operations in general, in addition to perfect classical channels connecting classical input and output of two local operations considered in LOCC.

We show that the set of joint quantum operations with CC* respecting the causal order is equivalent to LOCC. Equivalently, we show that a deterministic joint quantum operation is in LOCC if and only if it has a decomposition in Eq.~(4), there is a partial order over local quantum operations, and $p\left(i_1,\cdots , i'_N|o_1, \cdots, o'_N \right)$ in Eq.~(\ref{eq:LOCC*0}) satisfies the no-signaling condition with respect to the partial order. 
This property of deterministic joint quantum operations in the form of Eq.~(\ref{eq:LOCC*0}) suggests that LOCC is a set of all possible deterministic joint quantum operations implementable by local operations and {\it classical communication respecting causal order}. 

In the next subsection, we give a formal definition of CC* respecting the causal order in terms of the no-signaling condition. Then, in the last subsection, we give a proof for our proposition.

\subsection{CC* respecting causal order}

Suppose there is a  strict partial order ``$\prec$'' on the set 
of  local operations $\{\mathcal{A}^{(k)}_{o_k|i_k}\}_{k=1}^N \cup 
\{\mathcal{B}^{(k)}_{o_k|i_k}\}_{k=1}^N$, where  ``strict'' means that 
neither $\mathcal{A}^{(k)}_{o_k|i_k}\prec 
\mathcal{A}^{(k)}_{o_k|i_k}$ nor $\mathcal{B}^{(k)}_{o_k|i_k} \prec 
\mathcal{B}^{(k)}_{o_k|i_k}$
holds. We assume that the order of local operations in each laboratory is fixed and the local operations are performed in the increasing order of $k$, which is the assumption of local temporal ordering introduced in the main text.

We consider that this partial order $\prec$ represents the causal order of the spacetime coordinates of local operations.  For example, $\mathcal{A}^{(k)}_{o_k|i_k} \prec \mathcal{B}^{(l)}_{o'_l|i'_l}$ 
means that $\mathcal{B}^{(l)}_{o'_l|i'_l}$ is performed after $\mathcal{A}^{(k)}_{o_k|i_k}$ had been performed.  Neither $\mathcal{A}^{(k)}_{o_k|i_k} \prec \mathcal{B}^{(l)}_{o'_l|i'_l}$ nor $ \mathcal{B}^{(l)}_{o'_l|i'_l}  \prec \mathcal{A}^{(k)}_{o_k|i_k} $ holds if the spacetime coordinates of  $\mathcal{A}^{(k)}_{o_k|i_k}$ and $\mathcal{B}^{(l)}_{o'_l|i'_l}$ are spacelike separated.  Due to the assumption of local temporal ordering, $\mathcal{A}^{(k+1)}_{o_{k+1}|i_{k+1}}$ is performed after $\mathcal{A}^{(k)}_{o_{k}|i_{k}}$, thus ``$\prec$''  satisfies  $\mathcal{A}^{(k)}_{o_{k}|i_{k}} \prec \mathcal{A}^{(k+1)}_{o_{k+1}|i_{k+1}}$ for all $k=1, \cdots ,N-1$. Similarly, it satisfies $\mathcal{B}^{(k)}_{o'_{k}|i'_{k}} \prec \mathcal{B}^{(k+1)}_{o'_{k+1}|i'_{k+1}}$ for all $k=1, \cdots ,N-1$. 

We define a set ${past}(\mathfrak{I})$ representing a set of all ``{\it past}'' outputs for a set of inputs $\mathfrak{I}$, where  $\mathfrak{I}$ is a subset of  $\{i_k\}_{k=1}^N \cup \{i'_k\}_{k=1}^N$. Formally, by introducing a set of local  operations $Op\left(\mathfrak{I}\right)$ corresponding to $\mathfrak{I}$ as 
\begin{equation}
 Op\left(\mathfrak{I}\right):=\left\{ \mathcal{A}^{(k)}_{o_{k}|i_{k}} 
| i_k \in \mathfrak{I}\right\} \cup \left\{ \mathcal{B}^{(k)}_{o'_{k}|i'_{k}} 
| i'_k \in \mathfrak{I}\right\},
\end{equation}
${past}(\mathfrak{I})$ 
is defined as 
\begin{eqnarray}
 past\left(\mathfrak{I}\right):=&\{o_k | \exists \chi \in 
 Op\left(\mathfrak{I}\right), \mathcal{A}^{(k)}_{o_{k}|i_{k}} \prec 
 \chi \} \nonumber\\&\cup \{o'_k | \exists \chi \in 
 Op\left(\mathfrak{I}\right), \mathcal{B}^{(k)}_{o'_{k}|i'_{k}} \prec \chi \},
\end{eqnarray}
where $\chi$ represents any element of $Op\left(\mathfrak{I}\right)$, that is, there is $l$ such that $\chi=\mathcal{A}^{(l)}_{o_{l}|i_{l}}$ or $\chi=\mathcal{B}^{(l)}_{o'_{l}|i'_{l}}$. 
Note that  $o_k$ is not in $past \left(\{i_k\}\right)$ by definition.

Next we define CC* respecting the causal order. The conditional probability distribution $p\left(i_1, \cdots,  i'_N|o_1, \cdots, o'_N\right)$ in Eq.~(\ref{eq:LOCC*0}) is said to be CC* respecting causal order if there is a partial order $\prec$ on the set of local operations $\{\mathcal{A}^{(k)}_{o_k|i_k}\}_{k=1}^N \cup \{\mathcal{B}^{(k)}_{o_k|i_k}\}_{k=1}^N$ satisfying the following conditions:
\begin{itemize}
 \item For all $k$, {  $\mathcal{A}^{(k)}_{o_{k}|i_{k}} \prec 
 \mathcal{A}^{(k+1)}_{o_{k+1}|i_{k+1}}$ and  $\mathcal{B}^{(k)}_{o'_{k}|i'_{k}} \prec \mathcal{B}^{(k+1)}_{o'_{k+1}|i'_{k+1}}$}.
\item For all $k$ and $l$, 
\begin{eqnarray}\label{App Eq p temporal ordering}
 p\left(i_1, \cdots, i_k, i'_1, \cdots i'_l|o_1, 
\cdots, o_N, o'_1, \cdots, o'_N\right)\nonumber\\= p \left(i_1, \cdots, i_k, i'_1, 
\cdots i'_l| past \left ( \{ i_a \}_{ a=1 }^k \cup \{ i'_b \}_{ b=1 }^l  \right ) \right),\nonumber\\
\end{eqnarray}
where
\begin{eqnarray}\label{App Eq definition of p}
 p\left(i_1, \cdots, i_k, i'_1, \cdots i'_l|o_1, 
\cdots, o_N, o'_1, \cdots, o'_N\right)\nonumber\\
= \sum _{\substack{i_{k+1},\cdots i_N\\i'_{l+1} 
\cdots i'_N}}
p\left(i_1, \cdots, i_N, i'_1, \cdots i'_N|o_1,
\cdots, o_N, o'_1, \cdots, o'_N\right).\nonumber\\
\end{eqnarray}

\end{itemize} 
Hence the conditional probability distribution in Eq.~(\ref{App Eq p temporal ordering}) can be regarded as a classical channel satisfying the no-signaling condition among multiple spacetime coordinates, namely, outputs of the classical channel never depend on inputs that are not in the past of the outputs. Note that an input $i_k$ for a local operation $\mathcal{A}^{(k)}_{o_k|i_k}$ is an output of the classical channel. 


\subsection{A necessary and sufficient condition for LOCC* to be LOCC}

We present the main proposition regarding the relationship between LOCC and causal order.
\begin{proposition}
A deterministic joint quantum operation is in LOCC, if and only if it has a decomposition in 
the form of Eq.~(\ref{eq:LOCC*0}) with $p\left(i_1, \cdots,  i'_N|o_1, \cdots, o'_N\right)$ respecting causal order.
\end{proposition}
From the proposition, we immediately derive the following corollary about LOCC*.
\begin{corollary}
 A deterministic joint quantum operation is in LOCC* and not in LOCC, if and only if 
 it has a decomposition in the form of Eq.~(\ref{eq:LOCC*0}) with $p\left(i_1, \cdots,  i'_N|o_1, \cdots, o'_N\right)$ and for all such decompositions,  
$p\left(i_1, \cdots,  i'_N|o_1, \cdots, o'_N\right)$ does not respect causal order.
\end{corollary}

Since the proposition is a necessary and sufficient condition for LOCC, it gives a new characterization of LOCC in terms of $p\left(i_1, \cdots,  i'_N|o_1, \cdots,  o'_N\right)$ and similarly, the corollary gives a new characterization of a deterministic joint quantum operation in LOCC* and not in LOCC, which is nothing but a non-LOCC separable quantum operation, in terms of causal order.  

\vspace{0.2cm}
\noindent{{\bf Proof}}

We provide a proof of the main proposition in the remaining part of Appendix A.
Since the conventional definition of LOCC in Eq.~(\ref{eq:LOCC2}) using the totally ordered local operations immediately gives a decomposition with $p\left(i_1, \cdots,  i'_N|o_1, \cdots, o'_N\right)$ respecting causal order, the ``only if'' part is trivial.  Hence we concentrate on proving the ``if'' part of the proposition. 

This proof consists of two steps.   In the first step, we show that for given local operations $\{\mathcal{A}^{(k)}_{o_k|i_k}\}_{k=1}^N \cup \{\mathcal{B}^{(k)}_{o_k|i_k}\}_{k=1}^N$ and $p\left(i_1, \cdots,  i'_N|o_1, \cdots, o'_N\right)$ respecting causal order, representing a deterministic joint quantum operation given by Eq.~(\ref{eq:LOCC*0}), it is always possible to construct sets of {\it totally} ordered local operations $\left\{\mathcal{C}^{(k)}_{\mathbb{O}_{f(k,0)}|\mathbb{J}_{f(k,0)}}\right\}_{k=1}^N$ and $\left\{\mathcal{D}^{(k)}_{\mathbb{O}_{f(k,1)}|\mathbb{J}_{f(k,1)}}\right\}_{k=1}^N$, where the input $\mathbb{J}_l$  of a local operation depends only on the output $\mathbb{O}_{l-1}$ of the previous local operation, representing the deterministic joint quantum operation by choosing appropriate noisy classical channels $Q(\mathbb{J}_1)$ and $\left\{Q(\mathbb{J}_k|\mathbb{O}_{k-1})\right\}_{k=2}^{2N}$.     In the second step, we show that the deterministic joint quantum operation given by Eq.~(\ref{eq:LOCC*0}) can be represented by another pair of sets of totally ordered local operations $\left\{\mathcal{C'}^{(k)}_{\mathbb{O}_{f(k,0)}|\mathbb{J}_{f(k,0)}}\right\}_{k=1}^N$ and $\left\{\mathcal{D'}^{(k)}_{\mathbb{O}_{f(k,1)}|\mathbb{J}_{f(k,1)}}\right\}_{k=1}^N$, connected by perfect classical channels. The local operations obtained in the second step are not totally ordered {\it alternately}, however,  the standard form of LOCC given by Eq.~(\ref{eq:LOCC2}) is derived by slightly modifying the local operations.

\vspace{0.5cm}
\noindent{\it Step 1: Construction of totally ordered local operations connected by one-to-one noisy classical channels}
\vspace{0.5cm}

Since there is always a total order that preserves the structure of a given partial order on a finite set,
we define a map $f$ satisfying 
\begin{itemize}
 \item $f$ is a bijection from $\{1,\cdots, N\} \times \{0,1\}$ 
to $\{1,\cdots, 2N \}$.
\item $f(k,b)<f(k+1,b)$ for all $k$ and $b$. 
\item {  $\mathcal{A}^{(k)}_{o_k|i_k} \prec \mathcal{B}^{(l)}_{o'_l|i'_l} 
\Rightarrow f(k,0) < f(l,1)$, and $\mathcal{B}^{(k)}_{o'_k|i'_k} \prec \mathcal{A}^{(l)}_{o_l|i_l} 
       \Rightarrow f(k,1) < f(l,0)$} for all $k$ and $l$.
\end{itemize}
By means of $f$, we define a set of classical inputs $\{I_l\}_{l=1}^{2N}$ and a set of outputs $\{O_l\}_{l=1}^{2N}$ whose elements are given by
\begin{equation}\label{app eq def I O}
I_{f(k,0)}=i_k, \ I_{f(k,1)}=i'_k, \  O_{f(k,0)}=o_k, \  
 \mbox{and } 
       O_{f(k,1)}=o'_k 
\end{equation}
for all $k$. We further define a conditional probability distribution $q(I_1,\cdots 
I_{2N}|O_1,\cdots O_{2N})$ by 
\begin{align}\label{app eq def q}
 q(I_1,\cdots, 
I_{2N}|O_1,\cdots, O_{2N}):=\nonumber\\
 p(i_1,\cdots, i_N, i'_1, \cdots i'_N|o_1,\cdots, o_N, o'_1, \cdots, o'_N),
\end{align}
where $I_l$ and $O_l$ are related to $i_k$, $i'_k$, $o_k$, and $o'_k$ by 
Eq. (\ref{app eq def I O}). Thus $q( \cdots | \cdots)$ is another representation of $p(\cdots |\cdots )$ in terms of
the newly defined classical inputs and outputs $\{I_l \}_{l=1}^{2N}$ and $\{O_l\}_{l=1}^{2N}$. Using $q(I_1,\cdots  I_{2N}|O_1,\cdots O_{2N})$,  the condition for $p\left(i_1, \cdots,  i'_N|o_1, \cdots, o'_N\right)$  to satisfy causal order given by Eq.~(\ref{App Eq p temporal ordering}) for all $k$ and $l$ is expressed by
\begin{equation}\label{app eq q temporal order}
 q(I_1,\cdots, 
I_{k}|O_1,\cdots, O_{2N}) = q(I_1,\cdots, 
I_{k}|O_1,\cdots, O_{k-1})
\end{equation}
for all $k$. 

Next, we define sets of all possible values of $I_k$ and $O_k$ denoted by $\mathcal{I}_k$ and $\mathcal{O}_k$, respectively,  for $k=1, \cdots 2N$.  We represent variables whose variations are over $\prod_{j=1}^{k}\mathcal{I}_j \times \prod_{j=1}^{k-1}\mathcal{O}_j $ and $\prod_{j=1}^{k}\mathcal{I}_j \times \prod_{j=1}^{k}\mathcal{O}_j $ 
by $\mathbb{J}_k$ and $\mathbb{O}_k$, respectively.    
$\{ \mathbb{J}_k\}_{k=1}^{2N}$ and  $\{ \mathbb{O}_k \}_{k=1}^{2N}$ are 
classical inputs and outputs of new local operations, which will be introduced latter.  We define a function $Q(\mathbb{J}_k|\mathbb{O}_{k-1})$ representing a classical channel linking between $\mathbb{J}_k$ and $\mathbb{O}_{k-1}$ as
\begin{align}
 Q(\mathbb{J}_k|\mathbb{O}_{k-1}) := &
 \frac{q(I_1,\cdots,I_k|O_1,\cdots,O_{k-1})}{q(I_1,\cdots,I_{k-1}|O_1,\cdots, O_{k-2})} 
 \nonumber\\
 &\cdot \delta \left(\mathbb{J}_k[1,k-1],  
 \mathbb{O}_{k-1}[1,k-1]\right) \nonumber \\
 &\cdot \delta \left(\mathbb{J}_k[k+1,2k-1], 
 \mathbb{O}_{k-1}[k,2k-2]\right), \nonumber\\
 &\quad (2 \le k \le 2N) \label{app eq Q 1}\\
  Q(\mathbb{J}_1) :=& q(I_1), \label{app eq Q 2}
\end{align}
where $\delta(x,y)$ is the Kronecker delta, $\mathbb{J}_k=(I_1,\cdots I_{k}, O_1, \cdots, O_{k-1})$ for $2 
\le k \le 2N$, $\mathbb{J}_1=I_1$, and $\mathbb{J}_k[l,m]$ is a vector 
consisting of the partial entries of $\mathbb{J}_k$, from the $l$-th entry through the $m$-th entry for $l<m$.

It is easy to check that $Q(\mathbb{J}_k|\mathbb{O}_{k-1})$ satisfies
\begin{align}
&\sum _{\mathbb{J}_1,\cdots, \mathbb{J}_{2N}} Q(\mathbb{J}_1) \cdot \delta(\mathbb{J}_1,\mathbb{O}_1[1])\nonumber\\
&\qquad\prod 
 _{k=2}^{2N}Q(\mathbb{J}_k|\mathbb{O}_{k-1})\cdot \delta 
 (\mathbb{J}_k,\mathbb{O}_k[1,2k-1])  \nonumber \\
=& q(I_1^{(2N)},\cdots I_{2N}^{(2N)}|O_1^{(2N)}, \cdots, O_{2N}^{(2N)}) 
 \cdot 
 \nonumber \\
& \qquad  \prod _{k=2}^{2N}\delta \Big( (I_1^{(k)},\cdots, I_{k-1}^{(k)}, 
 O_1^{(k)},\cdots,  O_{k-1}^{(k)}),\nonumber\\&\qquad
  (I_1^{(k-1)},\cdots, I_{k-1}^{(k-1)}, O_1^{(k-1)},\cdots, 
 O_{k-1}^{(k-1)}) \Big),
\label{eqn:Qrelation}
\end{align}
where $I_l^{(k)}$ and $O_l^{(k)}$ are given by $\mathbb{O}_k=(I_1^{(k)},\cdots, 
I_{k}^{(k)}, O_1^{(k)}, \cdots O_{k}^{(k)})$, and $\mathbb{O}_k[l]$ denotes the $l$-th entry of 
$\mathbb{O}_k$.  Eq.~(\ref{eqn:Qrelation}) indicates that $Q(\mathbb{J}_1) \prod 
 _{k=2}^{2N}Q(\mathbb{J}_k|\mathbb{O}_{k-1})$ and $q(I_1^{(2N)},\cdots I_{2N}^{(2N)}|O_1^{(2N)}, \cdots, O_{2N}^{(2N)})$ represent the same classical channel if $\mathbb{J}_k=\mathbb{O}_k[1,k-1]$, $\mathbb{O}_k[1,k-1]=\mathbb{O}_{k-1}[1,k-1]$, and $\mathbb{O}_k[k+1,2k]=\mathbb{O}_{k-1}[k,2k-2]$ are satisfied.

We define a local operation performed in Alice's laboratory for $\mathbb{J}_{f(k,0)}$ and $\mathbb{O}_{f(k,0)}$
denoted by
$\mathcal{C}^{(k)}_{\mathbb{O}_{f(k,0)}|\mathbb{J}_{f(k,0)}}$ as 
\begin{align}
 \mathcal{C}^{(k)}_{\mathbb{O}_{f(k,0)}|\mathbb{J}_{f(k,0)}} &:=\delta \left (\mathbb{J}_{f(k,0)}, 
 \mathbb{O}_{f(k,0)}[1,2f(k,0)-1] \right )\nonumber\\&\qquad \cdot 
 \mathcal{A}^{(k)}_{\mathbb{O}_{f(k,0)}[2f(k,0)]|\mathbb{J}_{f(k,0)}[f(k,0)]}, \label{app eq def mathcal C 1}
\end{align}
where $\mathbb{O}_{1}[1,1]:=\mathbb{O}_{1}[1]$ for $f(k,0)=1$ and  $\mathbb{J}_{k}[l]$ represents the $l$-th entry  of $\mathbb{J}_{k}$.
Similarly, we define a local operation performed in Bob's laboratory denoted by
$\mathcal{D}^{(k)}_{\mathbb{O}_{f(k,1)}|\mathbb{J}_{f(k,1)}}$ as 
\begin{align}
 \mathcal{D}^{(k)}_{\mathbb{O}_{f(k,1)}|\mathbb{J}_{f(k,1)}} :=&\delta \left (\mathbb{J}_{f(k,1)}, 
 \mathbb{O}_{f(k,1)}[1,2f(k,1)-1] \right ) \nonumber\\&\qquad\cdot 
 \mathcal{B}^{(k)}_{\mathbb{O}_{f(k,1)}[2f(k,1)]|\mathbb{J}_{f(k,1)}[f(k,1)]}. \label{app eq def mathcal D 1}
\end{align}

By linking classical outputs and inputs of local operations
$\mathcal{C}^{(k)}_{\mathbb{O}_{f(k,0)}|\mathbb{J}_{f(k,0)}}$ and 
$\mathcal{D}^{(k)}_{\mathbb{O}_{f(k,1)}|\mathbb{J}_{f(k,1)}}$ satisfying
the total order introduced by  the function $f$, we derive a representation of the deterministic joint quantum operation given by Eq.~(\ref{eq:LOCC*0}) as follows:
\begin{align}
& \sum_{\substack{\mathbb{J}_1,\cdots ,\mathbb{J}_{2N}\\\mathbb{O}_1,\cdots , 
 \mathbb{O}_{2N}}}Q(\mathbb{J}_1)\cdot \left(\prod 
 _{k=2}^{2N}Q(\mathbb{J}_k|\mathbb{O}_{k-1}) \right) \nonumber\\&
\qquad \mathcal{C}^{(N)}_{\mathbb{O}_{f(N,0)}|\mathbb{J}_{f(N,0)}}\circ \cdots \circ
\mathcal{C}^{(1)}_{\mathbb{O}_{f(1,0)}|\mathbb{J}_{f(1,0)}} \nonumber\\&\qquad\otimes \mathcal{D}^{(N)}_{\mathbb{O}_{f(N,1)}|\mathbb{J}_{f(N,1)}}\circ \cdots \circ
\mathcal{D}^{(1)}_{\mathbb{O}_{f(1,1)}|\mathbb{J}_{f(1,1)}} \nonumber \\
=&\sum_{\substack{\mathbb{J}_1,\cdots ,\mathbb{J}_{2N}\\\mathbb{O}_1,\cdots , 
 \mathbb{O}_{2N}}}Q(\mathbb{J}_1) \cdot \delta\left( 
 \mathbb{J}_1|\mathbb{O}_1[1]\right)  \nonumber\\&\qquad\cdot \Big(\prod 
 _{k=2}^{2N}Q(\mathbb{J}_k|\mathbb{O}_{k-1})\cdot  \delta\left(\mathbb{J}_{k},\mathbb{O}_{k}[1,2k-1]\right)\Big)  \nonumber \\
& \qquad \mathcal{A}^{(N)}_{\mathbb{O}_{f(N,0)}[2f(N,0)]|\mathbb{J}_{f(N,0)}[f(N,0)]}\circ \cdots\nonumber\\&\qquad\cdots \circ
\mathcal{A}^{(1)}_{\mathbb{O}_{f(1,0)}[2f(1,0)]|\mathbb{J}_{f(1,0)}[f(1,0)]} \nonumber\\
&\qquad \otimes \mathcal{B}^{(N)}_{\mathbb{O}_{f(N,1)}[2f(N,1)]|\mathbb{J}_{f(N,1)}[f(N,1)]}\circ \cdots\nonumber\\&\qquad\cdots \circ
\mathcal{B}^{(1)}_{\mathbb{O}_{f(1,1)}[2f(1,1)]|\mathbb{J}_{f(1,1)}[f(1,1)]} 
\nonumber \\
=&\sum_{\mathbb{O}_1,\cdots , 
 \mathbb{O}_{2N}}
q(I_1^{(2N)},\cdots I_{2N}^{(2N)}|O_1^{(2N)}, \cdots, O_{2N}^{(2N)}) 
 \cdot 
 \nonumber \\
& \qquad  \prod _{k=2}^{2N}\delta \Big( (I_1^{(k)},\cdots, I_{k-1}^{(k)}, 
 O_1^{(k)},\cdots, 
 O_{k-1}^{(k)}),\nonumber\\&\qquad(I_1^{(k-1)},\cdots, I_{k-1}^{(k-1)}, O_1^{(k-1)},\cdots, 
 O_{k-1}^{(k-1)}) \Big)
 \nonumber \\
& \qquad \mathcal{A}^{(N)}_{\mathbb{O}_{f(N,0)}[2f(N,0)]|\mathbb{O}_{f(N,0)}[f(N,0)]}\circ \cdots\nonumber\\&\qquad\cdots \circ
\mathcal{A}^{(1)}_{\mathbb{O}_{f(1,0)}[2f(1,0)]|\mathbb{O}_{f(1,0)}[f(1,0)]} \nonumber\\
&\qquad \otimes \mathcal{B}^{(N)}_{\mathbb{O}_{f(N,1)}[2f(N,1)]|\mathbb{O}_{f(N,1)}[f(N,1)]}\circ \cdots \nonumber\\&\qquad\cdots\circ
\mathcal{B}^{(1)}_{\mathbb{O}_{f(1,1)}[2f(1,1)]|\mathbb{O}_{f(1,1)}[f(1,1)]} 
\nonumber\\
=&\sum_{\mathbb{O}_{2N}}
q(I_1^{(2N)},\cdots I_{2N}^{(2N)}|O_1^{(2N)}, \cdots, O_{2N}^{(2N)}) 
 \cdot 
 \nonumber \\
& \qquad \mathcal{A}^{(N)}_{\mathbb{O}_{2N}[N+f(N,0)]|\mathbb{O}_{2N}[f(N,0)]}\circ \cdots\nonumber\\&\qquad\cdots \circ
\mathcal{A}^{(1)}_{1,\mathbb{O}_{2N}[N+f(1,0)]|\mathbb{O}_{2N}[f(1,0)]} \nonumber\\
&\qquad \otimes \mathcal{B}^{(N)}_{N,\mathbb{O}_{2N}[N+f(N,1)]|\mathbb{O}_{2N}[f(N,1)]}\circ \cdots \nonumber\\&\qquad\cdots\circ
\mathcal{B}^{(1)}_{\mathbb{O}_{2N}[N+f(1,1)]|\mathbb{O}_{2N}[f(1,1)]} 
\nonumber\\
=& \sum_{\substack{i_1,\cdots,i_N\\i'_1,\cdots, 
 i'_N\\o_1,\cdots,o_N\\o'_1,\cdots,o'_N}}p\left(i_1,\cdots,i_N,i'_1,\cdots, 
 i'_N|o_1,\cdots,o_N,o'_1,\cdots,o'_N\right) \nonumber \\
& \qquad \mathcal{A}^{(N)}_{o_N|i_N}\circ \cdots \circ
\mathcal{A}^{(1)}_{o_1|i_1} \otimes \mathcal{B}^{(N)}_{o'_N|i'_N}\circ \cdots \circ
\mathcal{B}^{(1)}_{o'_1|i'_1}, \label{app eq longest}
\end{align}
 where we used Eqs.~(\ref{app eq def mathcal C 1}) and (\ref{app eq def 
 mathcal D 1}) in the first equality, Eqs.~(\ref{app eq Q 1}) and 
 (\ref{app eq Q 2}) in the second equality, and Eq.~(\ref{app eq def q}) in 
 the fourth equality.
Note that in the first line of Eq.~(\ref{app eq longest}), 
classical input $\mathbb{J}_k$ of a local operation only depends on classical 
output $\mathbb{O}_{k-1}$  of the previous local operation in $Q(\mathbb{J}_k|\mathbb{O}_{k-1})$.

\vspace{0.5cm}
\noindent{\it Step 2: Modification of noisy classical channels into perfect classical channels}
\vspace{0.5cm}

In this step, we define new local operations performed in Alice's and Bob's laboratories denoted by
$\mathcal{C'}^{(k)}_{\mathbb{O}_{f(k,0)}|\mathbb{O}_{f(k,0)-1}}$ and 
$\mathcal{D'}^{(k)}_{\mathbb{O}_{f(k,1)}|\mathbb{O}_{f(k,1)-1}}$, respectively,
and show that the first line of Eq.~(\ref{app eq longest}) can be transformed to
the standard form of LOCC using these local operations. 
We define
$\mathcal{C'}^{(k)}_{\mathbb{O}_{f(k,0)}|\mathbb{O}_{f(k,0)-1}}$ and 
$\mathcal{D'}^{(k)}_{\mathbb{O}_{f(k,1)}|\mathbb{O}_{f(k,1)-1}}$ 
as 
\begin{align}
 \mathcal{C'}^{(k)}_{\mathbb{O}_{f(k,0)}|\mathbb{O}_{f(k,0)-1}}:= 
 \sum_{\mathbb{J}_{f(k,0)}}Q(\mathbb{J}_{f(k,0)}|\mathbb{O}_{f(k,0)-1})
 \mathcal{C}^{(k)}_{\mathbb{O}_{f(k,0)}|\mathbb{J}_{f(k,0)}}, \nonumber \\
 \mathcal{D'}^{(k)}_{\mathbb{O}_{f(k,1)}|\mathbb{O}_{f(k,1)-1}}:= 
 \sum_{\mathbb{J}_{f(k,1)}}Q(\mathbb{J}_{f(k,1)}|\mathbb{O}_{f(k,1)-1})
 \mathcal{D}^{(k)}_{\mathbb{O}_{f(k,1)}|\mathbb{J}_{f(k,1)}}, \label{app eq def mathcal C' D'}
\end{align}
where we have used $Q(\mathbb{J}_{1}|\mathbb{O}_{0}):=Q(\mathbb{J}_{1})$.  By means of Eqs.~(\ref{app eq Q 1}) and (\ref{app eq Q 2}), we obtain a relation
\begin{align}
   &\sum_{\mathbb{O}_{f(k,0)}}\mathcal{C'}^{(k)}_{\mathbb{O}_{f(k,0)}|\mathbb{O}_{f(k,0)-1}} \nonumber\\&= \sum_{I_{f(k,0)},O_{f(k,0)}}\frac{q\left(I_1,\cdots , I_{f(k,0)}|O_1, 
 \cdots, O_{f(k,0)-1}\right)}{q\left(I_1,\cdots , I_{f(k,0)-1}|O_1, 
 \cdots, 
 O_{f(k,0)-2}\right)}\nonumber\\&\qquad\qquad\qquad\qquad\mathcal{A}^{(k)}_{O_{f(k,0)}|\left(I_{f(k,0)}\right)},\label{app  eq tp}
\end{align}
where $I_{f(N,0)}$  and $O_{f(N,0)}$ are the $f(N,0)$-th 
and $2f(N,0)$-th entries of $\mathbb{O}_{f(N,0)}$, respectively, and 
$\{I_l \}_{l=1}^{f(k,0)-1}$ and $\{O_l \}_{l=1}^{f(k,0)-1}$ are given 
by $\mathbb{O}_{f(k,0)-1}=\left(I_1, \cdots,  I_{f(k,0)-1}|O_1, \cdots, O_{f(k,0)-1} \right)$.  Eq.~(\ref{app eq q temporal order}) represents the property of $q(I_1,\cdots, I_{2N}|O_1,\cdots, O_{2N})$ satisfying causal order and Eq.~(\ref{app eq tp}) guarantees that  
$\{ 
\mathcal{C'}^{(k)}_{\mathbb{O}_{f(k,0)}|\mathbb{O}_{f(k,0)-1}}\}_{\mathbb{O}_{f(k,0)}}
$ is indeed a quantum instrument for all $k$. Similarly,  $\{ 
\mathcal{D'}^{(k)}_{\mathbb{O}_{f(k,1)}|\mathbb{O}_{f(k,1)-1}}\}_{\mathbb{O}_{f(k,1)}}
$ is also shown to be a quantum instrument for all $k$. 

Now the deterministic joint quantum operation given in the form of Eq.~(\ref{eq:LOCC*0}) can be represented in terms of quantum instruments $\{ 
\mathcal{C'}^{(k)}_{\mathbb{O}_{f(k,0)}|\mathbb{O}_{f(k,0)-1}}\}_{\mathbb{O}_{f(k,0)}}
$ and $\{ 
\mathcal{D'}^{(k)}_{\mathbb{O}_{f(k,1)}|\mathbb{O}_{f(k,1)-1}}\}_{\mathbb{O}_{f(k,1)}}
$ as
\begin{align}
 & \sum_{\substack{i_1,\cdots,i_N\\i'_1,\cdots, 
 i'_N\\o_1,\cdots,o_N\\o'_1,\cdots,o'_N}}p\left(i_1,\cdots,i_N,i'_1,\cdots, 
 i'_N|o_1,\cdots,o_N,o'_1,\cdots,o'_N\right) \nonumber \\
& \qquad \mathcal{A}^{(N)}_{o_N|i_N}\circ \cdots \circ
\mathcal{A}^{(1)}_{o_1|i_1} \otimes \mathcal{B}^{(N)}_{o'_N|i'_N}\circ \cdots \circ
\mathcal{B}^{(1)}_{o'_1|i'_1}, \nonumber \\
=&
\sum_{\mathbb{O}_1,\cdots , \mathbb{O}_{2N}}
 \mathcal{C'}^{(N)}_{\mathbb{O}_{f(N,0)}|\mathbb{O}_{f(N,0)-1}}\circ \cdots \circ
\mathcal{C'}^{(1)}_{\mathbb{O}_{f(1,0)}|\mathbb{O}_{f(1,0)-1}} \nonumber\\&\qquad\otimes \mathcal{D'}^{(N)}_{\mathbb{O}_{f(N,1)}|\mathbb{O}_{f(N,1)-1}}\circ \cdots \circ
\mathcal{D'}^{(1)}_{\mathbb{O}_{f(1,1)}|\mathbb{O}_{f(1,1)-1}}. 
\label{app eq final}
\end{align}
by using Eq.~(\ref{app eq def mathcal C' D'}).  
The right hand side of Eq.~(\ref{app eq final}) is almost 
in the standard form of LOCC.  The only case of Eq.~(\ref{app eq final}) not fitting in LOCC is that where Alice (Bob) performs two local operations successively.   This case is absorbed in LOCC in the following manner.  Suppose Alice 
successively performs two local operations $\mathcal{C'}^{(k-1)}_{\mathbb{O}_{f(k-1,0)}|\mathbb{O}_{f(k-1,0)-1}}$ and $\mathcal{C'}^{(k)}_{\mathbb{O}_{f(k,0)}|\mathbb {O}_{f(k,0)-1}}$ where $f(k,0)=f(k-1,0)+1$.  Since 
\begin{equation}
\left \{\sum_{\mathbb{O}_{f(k-1,0)}}\mathcal{C'}^{(k)}_{\mathbb{O}_{f(k,0)}|\mathbb
{O}_{f(k,0)-1}}\circ 
\mathcal{C'}^{(k-1)}_{\mathbb{O}_{f(k-1,0)}|\mathbb{O}_{f(k-1,0)-1}} \right 
\}_{\mathbb{O}_{f(k,0)}} \nonumber 
\end{equation}
is a quantum instrument representing a local operation performed in Alice's laboratory,  
we regard these successive local operations as a single local 
operation performed in Alice's laboratory.  Similarly, we combine successive local operations performed in Bob's laboratory into a single local operation.  By repeating this procedure, we can rewrite Eq.~(\ref{app eq final}) in the standard form of LOCC, where a sequence of local operations is performed alternatively in Alice's and Bob's laboratories.  Hence, we can conclude that Eq.~(\ref{app eq final}) reduces to a standard decomposition of LOCC given in the form of Eq.~(\ref{eq:LOCC2}).
Therefore,  the ``{\it if}'' part of the proposition is proven.

\renewcommand{\theequation}{B\arabic{equation}}
\section{Analysis of LOCC** in terms of the CJ representation}

In this appendix, we analyze LOCC**, a set of deterministic joint quantum operations where local operations are linked by a classical process matrix. Since the process matrix is formulated by using  the Choi-Jamio\l{}kowski (CJ) representation of quantum operations, we first review CJ representations and classical process matrix, next give the formal definition of LOCC**, and then show that LOCC** can be reduced to a probabilistic mixture of one-way LOCC in a bipartite setting.

\subsection{The Choi-Jamio\l{}kowski representation}

The Choi-Jamio\l{}kowski (CJ) representation is a way to represent general quantum operations given by quantum instruments as linear operators on a composite Hilbert space of an input Hilbert space and an output Hilbert space.   We denote a CJ representation of a quantum operation described by a map $\mathcal{M}$ by a linear operator $M$.  The Hilbert spaces of a quantum input system and  a quantum output system are denoted by $\mathcal{H}_{in}$ and $\mathcal{H}_{out}$, respectively.  $L(\mathcal{H})$ denotes the space of linear operators over a Hilbert space $\mathcal{H}$.   For a map representing an element of a quantum instrument $\mathcal{M}_{o|i}:L(\mathcal{H}_{in})\rightarrow L(\mathcal{H}_{out})$, where $i$ is an index of the classical input and $o$ is an index of the classical output,  the corresponding CJ operator $M_{o|i}\in L(\mathcal{H}_{in}\otimes\mathcal{H}_{out})$  is given by
\begin{equation}
M_{o|i}=\sum_{k,l}\ket{k}\bra{l}\otimes\mathcal{M}_{o|i}(\ket{k}\bra{l}),
\end{equation}
where $\{\ket{k}\}$ is the (fixed) computational basis on $\mathcal{H}_{in}$.  The state of a quantum output for a quantum input $\rho_{in}$ given by a linear map $\mathcal{M}_{o|i}$ is obtained by using the CJ operator $M_{o|i}$ as 
\begin{equation}
\mathcal{M}_{o|i}(\rho_{in})=\mathrm{tr}_{in}[M_{o|i}(\rho_{in}^T\otimes\mathbb{I}_{out})]
\end{equation}
where $\mathbb{I}_{out}$ is the identity operator on $\mathcal{H}_{out}$, and $\rho_{in}^T$ is the transposition of $\rho_{in}$ with respect to the computational basis.  Note that the output state does not depend on the choice of the computational basis.
$\mathcal{M}_{o|i}$ is completely positive (CP) if and only if $M_{o|i}$ is a positive semidefinite operator. 
$\sum_o\mathcal{M}_{o|i}$ is trace preserving (TP) if and only if $\mathrm{tr}_{out}[\sum_o M_{o|i}]=\mathbb{I}_{in}$.

\subsection*{Classical process matrix}

 A process matrix \cite{OFC} represents a higher order map transforming quantum operations to another quantum operation. We denote the set of all positive semi-definite operators on a Hilbert space $\mathcal{H}$ as $Pos(\mathcal{H})$ and the set of all CJ operators on $\mathcal{H}_{I}\otimes\mathcal{H}_{O}$ representing CPTP maps from $\mathcal{H}_I$
to $\mathcal{H}_O$ as $CPTP(\mathcal{H}_{I}:\mathcal{H}_{O}):=\{ Q \in Pos(\mathcal{H}_{I}\otimes\mathcal{H}_{O})|\mathrm{tr}_{O} Q =\mathbb{I}_I\}$.

To be consistent with quantum mechanics, it has been proven in \cite{OFC} that a process matrix linking two local operations represented by the CJ operators $M_A\in L(\mathcal{H}_{I_A}\otimes\mathcal{H}_{O_A})$ and $M_B\in L(\mathcal{H}_{I_B}\otimes\mathcal{H}_{O_B})$ as inputs, is a positive semi-definite operator $W \in Pos(\mathcal{H}_{I_A}\otimes\mathcal{H}_{O_A}\otimes\mathcal{H}_{I_B}\otimes\mathcal{H}_{O_B})$ and that satisfies
\begin{equation}
\mathrm{tr}\big[W(M_A^\mathrm{T}\otimes M_B^\mathrm{T})\big]=1,
\label{eq:process}
\end{equation}
for all $M_A\in CPTP(\mathcal{H}_{I_A}:\mathcal{H}_{O_A})$ and $M_B\in CPTP(\mathcal{H}_{I_B}:\mathcal{H}_{O_B})$, where $Q^\mathrm{T}$ represents the transposition of $Q$ with respect to the computational basis used in defining the CJ representation. 

A deterministic joint quantum operation $\mathcal{M}:L(\mathcal{H}_X\otimes\mathcal{M}_Y)\rightarrow L(\mathcal{H}_A\otimes\mathcal{H}_B)$ can be regarded as being obtained by transforming two local operations  $\mathcal{A}:L(\mathcal{H}_{I_A}\otimes\mathcal{H}_X)\rightarrow L(\mathcal{H}_{O_A}\otimes\mathcal{H}_A)$ and $\mathcal{B}:L(\mathcal{H}_{I_B}\otimes\mathcal{H}_Y)\rightarrow L(\mathcal{H}_{O_B}\otimes\mathcal{H}_B)$ by  a process matrix $W \in Pos(\mathcal{H}_{I_A}\otimes\mathcal{H}_{O_A}\otimes\mathcal{H}_{I_B}\otimes\mathcal{H}_{O_B})$  linking the Hilbert spaces $\mathcal{H}_{I_A}$, $\mathcal{H}_{O_A}$, $\mathcal{H}_{I_B}$ and $\mathcal{H}_{O_B}$.  The CJ operator $M \in L(\mathcal{H}_X\otimes\mathcal{H}_Y\otimes \mathcal{H}_A\otimes\mathcal{H}_B)$ of $\mathcal{M}$ obtained by transforming the CJ operators of local operations $A \in L(\mathcal{H}_{I_A}\otimes\mathcal{H}_X \otimes \mathcal{H}_{O_A}\otimes\mathcal{H}_A )$ and $B\in L(\mathcal{H}_{I_B}\otimes\mathcal{H}_Y \otimes \mathcal{H}_{O_B}\otimes\mathcal{H}_B )$ corresponding to $\mathcal{A}$ and $\mathcal{B}$, respectively, by $W$ satisfying Eq.~(\ref{eq:process})  is represented by 
\begin{equation}
M=\mathrm{tr}_{I_A,O_A,I_B,O_B}\big[W(A^{\mathrm{T}_{I_A,O_A}}\otimes B^{\mathrm{T}_{I_B,O_B}})\big],
\label{eq:LOCC**Wmatrix}
\end{equation}
where $Q^{\mathrm{T}_{I_X,O_X}}$ is the partial transposition of $Q$, taking the transposition only in terms of $\mathcal{H}_{I_X}$ and $\mathcal{H}_{O_X}$. 

Quantum communication from Alice to Bob or Bob to Alice can be described by $W$ if $W$ is equivalent to the CJ operator  representing a quantum channel  linking two causally ordered local operations performed by different parties. Moreover, process matrices can represent ``quantum communication'' without causal order, which is not implementable when the partial order of local operations is fixed but is not ruled out in the framework of quantum mechanics \cite{OFC}.    The restriction for process matrices given by Eq.~\eqref{eq:process} can be interpreted to represent a new kind of causality required for linking local operations in quantum mechanics.

 We consider a special class of process matrix where Alice and Bob can communicate only by ``classical'' states.  In this case, the states on the local input and output Hilbert spaces of $W$, $\mathcal{H}_{I_A}$, $\mathcal{H}_{O_A}$, $\mathcal{H}_{I_B}$ and $\mathcal{H}_{O_B}$, are restricted and are diagonal with respect to the computational basis.  In such a case,  the process matrix $W$ is called a {\it classical} process matrix, representing classical communication between the parties, and can be described by a conditional probability distribution $p(i_A,i_B|o_A,o_B)$ satisfying
\begin{eqnarray}
\sum_{i_A,i_B,o_A,o_B} p(i_A,i_B|o_A,o_B)A_{o_A|i_A}\otimes B_{o_B|i_B}\nonumber\\\quad\in CPTP(\mathcal{H}_X\otimes\mathcal{H}_Y:\mathcal{H}_A\otimes \mathcal{H}_B)
\label{eq:LOCC**}
\end{eqnarray}
for all quantum instruments $\{A_{o_A|i_A}\}_{o_A}$ and $\{B_{o_B|i_B} \}_{o_B}$, where $\{A_{o_A|i_A}\in Pos(\mathcal{H}_X\otimes \mathcal{H}_A)\}_{o_A}$ and $\{B_{o_B|i_B}\in Pos(\mathcal{H}_Y\otimes \mathcal{H}_B)\}_{o_B}$ satisfying $\sum_{o_A}\mathrm{tr}_{A}[A_{o_A|i_A}]=\mathbb{I}_X$ and $\sum_{o_B}\mathrm{tr}_{B}[B_{o_B|i_B}]=\mathbb{I}_Y$.
The proof is given  as follows.

\begin{proof}
 By denoting the computational basis of $\mathcal{H}_{X}$ as $\{\ket{x}_{X}\}$,  local operations $M_A\in CPTP(\mathcal{H}_{I_A}:\mathcal{H}_{O_A})$, $M_B\in CPTP(\mathcal{H}_{I_B}:\mathcal{H}_{O_B})$, and a classical process matrix described by a diagonal process matrix $W\in Pos(\mathcal{H}_{I_A}\otimes\mathcal{H}_{O_A}\otimes\mathcal{H}_{I_B}\otimes\mathcal{H}_{O_B})$ can be decomposed into
\begin{eqnarray}
M_A&=&\sum_{i_A,o_A}p_A(o_A|i_A)\ket{i_A,o_A}\bra{i_A,o_A}_{I_A,O_A},\\
M_B&=&\sum_{i_B,o_B}p_B(o_B|i_B)\ket{i_B,o_B}\bra{i_B,o_B}_{I_B,O_B},\\
W&=&\sum_{\substack{i_A,o_A\\i_B,o_B}}w(i_A,i_B,o_A,o_B)\nonumber\\&&\qquad\ket{i_A,i_B,o_A,o_B}\bra{i_A,i_B,o_A,o_B}_{I_A,I_B,O_A,O_B}, \nonumber\\\label{eq:diagonalW}
\end{eqnarray}
where $w(i_A,i_B,o_A,o_B)$ represents a diagonal element of $W$, and $p_A(o_A|i_A)$ and $p_B(o_B|i_B)$ are conditional probability distributions since $M_A$ and $M_B$ are CPTP maps.  In the following, we show that classical process matrices correspond to conditional probability distributions, namely, $w(i_A,i_B,o_A,o_B)$ must be a conditional probability distribution so that Eq.~\eqref{eq:process} holds.
The non-negativity of $W$ implies $w(i_A,i_B,o_A,o_B)\geq 0$ and the condition given by Eq.~\eqref{eq:process} is equivalent to 
\begin{equation}
\sum_{i_A,i_B,o_A,o_B} w(i_A,i_B,o_A,o_B)p_A(o_A|i_A)p_B(o_B|i_B)=1
\end{equation}
for all conditional probability distributions $p_A(o_A|i_A)$ and $p_B(o_B|i_B)$.  By choosing $p_A(o_A|i_A)=\delta_{o_A,a}$ and $p_B(o_B|i_B)=\delta_{o_B,b}$, we obtain $\sum_{i_A,i_B} w(i_A,i_B,a,b)=1$ for arbitrary $a$ and $b$.  Thus $w(i_A,i_B,o_A,o_B)$ can be represented by a conditional probability distribution conditioned by $o_A$ and $o_B$, and we define
\begin{equation}
p(i_A,i_B|o_A,o_B):=w(i_A,i_B,o_A,o_B).
\end{equation}

To satisfy \eqref{eq:process}, $p(i_A,i_B|o_A,o_B)$ satisfies
\begin{equation}
\sum_{i_A,i_B,o_A,o_B} p(i_A,i_B|o_A,o_B)p_A(o_A|i_A)p_B(o_B|i_B)=1
\label{eq:def1_CC**}
\end{equation}
for all conditional probability distributions $p_A(o_A|i_A)$ and $p_B(o_B|i_B)$. The condition given by Eq.~\eqref{eq:def1_CC**} is equivalent to that given by Eq.~\eqref{eq:LOCC**} if $p(i_A,i_B|o_A,o_B)$ is a conditional probability distribution.  This can be shown as follows.   By letting $dim(\mathcal{H}_X)=dim(\mathcal{H}_Y)=dim(\mathcal{H}_A)=dim(\mathcal{H}_B)=1$, $A_{o_A|i_A}=p(o_A|i_A)$ and $B_{o_B|i_B}=p(o_B|i_B)$, it is easy to confirm that Eq.~\eqref{eq:def1_CC**} holds if Eq.~\eqref{eq:LOCC**} holds.  To show the converse, we first check that a map decomposable in the form of Eq.~\eqref{eq:LOCC**} is completely positive.  This is also easily checked since every term in the summation is non-negative.   Then we show that a map decomposable in the form of Eq.~\eqref{eq:LOCC**} is trace preserving when Eq.~\eqref{eq:def1_CC**} holds in the following.
For any operator $\sigma\in L(\mathcal{H}_X\otimes\mathcal{H}_Y)$,
\begin{eqnarray}
&\mathrm{tr}_{A,B}\Big[\mathrm{tr}_{X,Y}\big[\sum_{\substack{i_A,i_B\\o_A,o_B}} p(i_A,i_B|o_A,o_B)\nonumber\\
&\qquad(A_{o_A|i_A}\otimes B_{o_B|i_B}) \sigma^\mathrm{T}\big]\Big]\\
=&\sum_{k,l}\lambda_{k,l}\sum_{\substack{i_A,i_B\\o_A,o_B}} p(i_A,i_B|o_A,o_B)\nonumber\\
&\qquad\mathrm{tr}_{A,X}\big[A_{o_A|i_A}\rho_k^\mathrm{T}\big]\mathrm{tr}_{B,Y}\big[B_{o_B|i_B}\rho_l^\mathrm{T}\big]\\
=&\sum_{k,l}\lambda_{k,l}=\mathrm{tr}[\sigma],
\end{eqnarray}
where $\sigma$ is decomposed as $\sigma=\sum_{k.l}\lambda_{k,l}\rho_k\otimes\rho_l$ by using density operators $\{\rho_k\}_k$ as a basis of the linear space of the operator.  Note that $\mathrm{tr}_{A,X}\big[A_{o_A|i_A}\rho_k^T\big]$ is a conditional probability distribution conditioned by $i_A$ since it satisfies $\sum_{o_A}\mathrm{tr}_{A,X}\big[A_{o_A|i_A}\rho_k^T\big]=\mathrm{tr}_{A,X}\big[\sum_{o_A}A_{o_A|i_A}\rho_k^T\big]=\mathrm{tr}_{X}\big[\rho_k^T\big]=1$.
\end{proof}

\subsection*{Formal definition of LOCC**}
Next, we define a set of deterministic joint quantum operations consisting of local operations linked by a classical process matrix, denoted by LOCC**. Local operations $A\in CPTP(\mathcal{H}_{I_A}\otimes\mathcal{H}_{X}:\mathcal{H}_{O_A}\otimes\mathcal{H}_{A})$ and $B\in CPTP(\mathcal{H}_{I_B}\otimes\mathcal{H}_{Y}:\mathcal{H}_{O_B}\otimes\mathcal{H}_{B})$ linked by a diagonal process matrix $W\in Pos(\mathcal{H}_{I_A}\otimes\mathcal{H}_{I_B}\otimes\mathcal{H}_{O_A}\otimes\mathcal{H}_{O_B})$ given in the form of Eq.~\eqref{eq:diagonalW} can be decomposed into
\begin{eqnarray}
A&=&\sum_{i_A,o_A}A_{o_A|i_A}\otimes\ket{i_A,o_A}\bra{i_A,o_A}_{I_A,O_A},\\
B&=&\sum_{i_B,o_B}B_{o_B|i_B}\otimes\ket{i_B,o_B}\bra{i_B,o_B}_{I_B,O_B},
\end{eqnarray}
where $\{A_{o_A|i_A}\in Pos(\mathcal{H}_X\otimes \mathcal{H}_A)\}_{o_A}$ and $\{B_{o_B|i_B}\in Pos(\mathcal{H}_Y\otimes \mathcal{H}_B)\}_{o_B}$ are the CJ operators of quantum instruments since $A$ and $B$ are CPTP maps.
By straightforward calculation, the CJ operator of a deterministic joint quantum operation is written as
\begin{eqnarray}
M&=&\mathrm{tr}_{I_A,O_A,I_B,O_B}\big[W(A^{\mathrm{T}_{I_A,O_A}}\otimes B^{\mathrm{T}_{I_B,O_B}})\big]
\\
&=&\sum_{\substack{i_A,i_B\\o_A,o_B}}p(i_A,i_B|o_A,o_B)A_{o_A|i_A}\otimes B_{o_B|i_B},
\label{eq:def1_LOCC*}
\end{eqnarray}
which shows the equivalence of the representations of $M$ given by Eq.~\eqref{eq:LOCC*} and Eq.~\eqref{eq:LOCC**Wmatrix} when the process matrix is diagonal/classical.    

LOCC** is defined by a set of deterministic joint quantum operations whose CJ operators can be represented by Eq.~\eqref{eq:def1_LOCC*}, or in the form of
\begin{equation}
\mathcal{M}=\sum_{i_A,i_B,o_A,o_B}p(i_A,i_B|o_A,o_B)\mathcal{A}_{o_A|i_A}\otimes\mathcal{B}_{o_B|i_B},
\end{equation}
where $p(i_A,i_B|o_A,o_B)$ satisfies Eq.~\eqref{eq:LOCC**}. By definition, LOCC* is a larger set than LOCC**  since the condition for $p(i_A,i_B|o_A,o_B)$ given by Eq.~\eqref{eq:LOCC**} should be satisfied for {\it all} local operations in LOCC** while it should be satisfied only for {\it some} local operations in LOCC*.   

\subsection*{Equivalence of LOCC** and one-way LOCC}
In  the following proposition, we prove that LOCC** is equivalent to a probability mixture of one-way LOCC in  a bipartite setting.  Thus,  CC* used in implementing a deterministic joint quantum operation in LOCC* but not in LOCC cannot be represented by a classical process matrix. 

\begin{proposition}
LOCC** is equivalent to a probability mixture of one-way LOCC in a bipartite setting.
\end{proposition}

\begin{proof}
In \cite{OFC}, it is shown that any classical process matrix is {\it causally separable}, i.e. the CJ operator of a classical process matrix $W$ can be decomposed into the form
\begin{equation}
W=qW_{A\rightarrow B}+(1-q)W_{B\rightarrow A},
\end{equation}
where $q\in [0,1]$, $W_{A\rightarrow B}\in Pos(\mathcal{H}_{I_A}\otimes\mathcal{H}_{O_A}\otimes\mathcal{H}_{I_B}\otimes\mathcal{H}_{O_B})$ is diagonal with respect to the computational basis and satisfies the conditions given by 
\begin{eqnarray}
W_{A\rightarrow B}&=&\mathbb{I}_{O_B}\otimes W_{I_A,O_A,I_B},\nonumber\\
\mathrm{tr}_{I_B}\left[W_{I_A,O_A,I_B}\right]&=&\mathbb{I}_{O_A}\otimes\rho_{I_A},\nonumber\\
\mathrm{tr}_{I_A}\left[\rho_{I_A}\right]&=&1,
\label{eq:AtoB}
\end{eqnarray}
and similar conditions are satisfied by $W_{B\rightarrow A}\in Pos(\mathcal{H}_{I_A}\otimes\mathcal{H}_{O_A}\otimes\mathcal{H}_{I_B}\otimes\mathcal{H}_{O_B})$.  

In \cite{Chiribella1},  it is proven that an operator $W_{A\rightarrow B}$ satisfying the conditions given by Eq.~\eqref{eq:AtoB}  but not necessarily being diagonal in the computational basis corresponds to a special type of process matrix called {\it quantum comb} where Alice's operation and Bob's operation are linked by quantum communication from Alice to Bob.   Thus a causally separable process can be interpreted as a probabilistic mixture of quantum communication from Alice to Bob and from Bob to Alice. When a causally separable process is classical (diagonal with respect to the computational basis), the process can be interpreted as a probabilistic mixture of classical communication from Alice to Bob and from Bob to Alice.

Let us denote the diagonal elements of $W_{A\rightarrow B}$ and $W_{B\rightarrow A}$ with respect to the computational basis by $p_{A\rightarrow B}(i_A,i_B|o_A,o_B)$ and $p_{B\rightarrow A}(i_A,i_B|o_A,o_B)$, respectively.   It is easy to verify that  $p_{A\rightarrow B}(i_A,i_B|o_A,o_B)$ and $p_{B\rightarrow A}(i_A,i_B|o_A,o_B)$ are conditional probability distributions since $W_{A\rightarrow B}$ and $W_{B\rightarrow A}$ correspond to classical process matrices.

Then $p(i_A,i_B|o_A,o_B)$ satisfying Eq.~\eqref{eq:LOCC**}
can be decomposed into $p(i_A,i_B|o_A,o_B)=qp_{A\rightarrow B}(i_A,i_B|o_A,o_B)+(1-q)p_{B\rightarrow A}(i_A,i_B|o_A,o_B)$.
Eq.~\eqref{eq:AtoB} implies that $p_{A\rightarrow B}(i_A,i_B|o_A,o_B)$ does not depend on $o_B$.  We define $p_{A\rightarrow B}(i_A,i_B|o_A):=p_{A\rightarrow B}(i_A,i_B|o_A,o_B)$.  The operator $W_{I_A,O_A,I_B}$ in Eq.~\eqref{eq:AtoB} is given by
\begin{eqnarray}
W_{I_A,O_A,I_B}=&\sum_{i_A,i_B,o_A}p_{A\rightarrow B}(i_A,i_B|o_A)\nonumber\\&\ket{i_A,o_A,i_B}\bra{i_A,o_A,i_B}_{I_A,O_A,I_B}.
\end{eqnarray}
Due to Eq.~\eqref{eq:AtoB}, $\sum_{i_B}p_{A\rightarrow B}(i_A,i_B|o_A)$ does not depend on $o_A$.  We define $p_{A\rightarrow B}(i_A):=\sum_{i_B}p_{A\rightarrow B}(i_A,i_B|o_A)$.
We further define a set $X=\{x|p_{A\rightarrow B}(x)\neq 0\}$ and
\begin{eqnarray}
p'_{A\rightarrow B}(i_B|o_A,i_A)&:=&\frac{p_{A\rightarrow B}(i_A,i_B|o_A)}{p_{A\rightarrow B}(i_A)}\,\,(i_A\in X)\nonumber\\
p'_{A\rightarrow B}(i_B|o_A,i_A)&:=&p'_{A\rightarrow B}(i_B)\,\,(i_A\notin X),
\end{eqnarray}
where $p'_{A\rightarrow B}(y)$ is an arbitrary probability distribution.  Note that $p'_{A\rightarrow B}(i_B|o_A,i_A)$ satisfies all the properties required for a conditional probability distribution. Thus, we have
\begin{equation}
p_{A\rightarrow B}(i_A,i_B|o_A,o_B)=p_{A\rightarrow B}(i_A)p'_{A\rightarrow B}(i_B|o_A,i_A),
\end{equation}
and similarly,
\begin{equation}
p_{B\rightarrow A}(i_A,i_B|o_A,o_B)=p_{B\rightarrow A}(i_B)p'_{B\rightarrow A}(i_A|o_B,i_B).
\end{equation}
Combining these results, a deterministic joint quantum operation $\mathcal{M}$ in LOCC** is given by
\begin{eqnarray}
\mathcal{M}&=&q\sum_{\substack{i_A,i_B\\o_A,o_B}}p_{A\rightarrow B}(i_A,i_B|o_A,o_B)\mathcal{A}_{o_A|i_A}\otimes\mathcal{B}_{o_B|i_B}\nonumber\\
&&+(1-q)\sum_{\substack{i_A,i_B\\o_A,o_B}}p_{B\rightarrow A}(i_A,i_B|o_A,o_B)\mathcal{A}_{o_A|i_A}\otimes\mathcal{B}_{o_B|i_B}.\nonumber\\
\end{eqnarray}
Since the first term can be represented by $q\sum_m\mathcal{A}_m\otimes\mathcal{B}_{|m}$ where 
\begin{eqnarray}
m&:=&(i_A,o_A),\nonumber\\
\mathcal{A}_{i_A,o_A}&:=&p_{A\rightarrow B}(i_A)\mathcal{A}_{o_A|i_A},\nonumber\\
\mathcal{B}_{|i_A,o_A}&:=&\sum_{i_B,o_B}p'_{A\rightarrow B}(i_B|o_A,i_A)\mathcal{B}_{o_B|i_B},
\end{eqnarray}
it represents an operation in one-way LOCC.   The deterministic joint quantum operation $\mathcal{M}$ in LOCC** is concluded to be a probability mixture of one-way LOCC.

\end{proof}

\nocite{*}
\bibliographystyle{eptcs}
\bibliography{generic}

\end{document}